\documentclass[11pt]{article}
\usepackage{titling}
\usepackage[OT1]{fontenc}
\usepackage{amsmath}
\DeclareMathOperator*{\argmax}{arg\,max}
\usepackage{mathrsfs}
\usepackage{amsfonts}
\usepackage{amssymb}
\usepackage{diagbox}
\usepackage{footnote}
\usepackage{graphicx}
\usepackage{float}
\usepackage{epstopdf}
\usepackage[margin=1.25in]{geometry}
\usepackage{setspace}
\usepackage{tikz}
\usepackage{pgfplots}
\usepackage[round]{natbib}
\usepackage{multirow}
\usepackage[toc,page]{appendix}
\usepackage{hyperref}
\usepackage{caption}
\usepackage{subcaption}
\usepackage{enumitem}
\usepackage{booktabs}
\usepackage{chapterbib}
\providecommand{\U}[1]{\protect\rule{.1in}{.1in}}

\hypersetup{
bookmarks=true,
	colorlinks=true,
	linkcolor=blue,
	citecolor=blue,
	filecolor=magenta,
	urlcolor=blue,
	breaklinks=true,
}

\newtheorem{assumption}{Assumption}[section]

\newtheorem{definition}{Definition}[section]

\newtheorem{lemma}{Lemma}[section]

\newtheorem{proposition}{Proposition}[section]

\newenvironment{proof}[1][Proof]{\noindent\textbf{#1.} }{\ \rule{0.5em}{0.5em}}

\begin{document}

\title{A Nonparametric Test of $m$th-degree Inverse Stochastic Dominance\thanks{The authors are grateful to Brendan K. Beare for the constructive comments. This work was supported by the National Natural Science Foundation of China [grant number 72103004].}}
\author{Hongyi Jiang \\China Center for Economic Research\\National School of Development\\ Peking University\\ hyjiang2017@nsd.pku.edu.cn\\
	\and Zhenting Sun\\China Center for Economic Research\\National School of Development\\ Peking University\\
	zhentingsun@nsd.pku.edu.cn \and Shiyun Hu \\China Center for Economic Research\\National School of Development\\ Peking University\\ hushiyun@pku.edu.cn\\}

\maketitle

\begin{abstract}
This paper proposes a nonparametric test for $m$th-degree inverse stochastic dominance which is a powerful tool for ranking distribution functions according to social welfare. We construct the test based on empirical process theory. The test is shown to be asymptotically size controlled and consistent. The good finite sample properties of the test are illustrated via Monte Carlo simulations. We apply our test to the inequality growth in the United Kingdom from 1995 to 2010.  
\end{abstract}
\bigskip


\newpage
\section{Introduction}
When we compare distribution functions according to social welfare, we usually rely on second-degree stochastic dominance \citep{atkinson1970measurement}. However, as pointed out by \citet{aaberge2021ranking}, second-degree dominance has limitations on comparing distribution functions that intersect. \citet{aaberge2021ranking} then propose a general approach to ranking intersecting distribution functions based on inverse stochastic dominance (ISD) introduced by \citet{muliere1989note}. 
\citet{aaberge2021ranking} consider two complementary sequences of inverse stochastic dominance criteria: Upward dominance and downward dominance. As demonstrated in \citet{aaberge2021ranking}, upward dominance aggregates the quantile function from below, so it places more emphasis on differences that occur in the lower part of the distribution; downward dominance aggregates the quantile function from above, so it places more emphasis on differences that occur in the upper part of the distribution.
\citet{aaberge2021ranking} also show that ISD of any degree can be given a social welfare interpretation. Though ISD plays an important role in welfare analysis, there are few statistical tools particularly designed for inference on ISD. \citet{aaberge2021ranking} develop distribution theory to test ISD and employ the approach of \citet{sverdrup1976significance} in testing applications.  
\citet{andreoli2018robust} develops a statistical testing approach of inverse stochastic dominance based on a finite set of abscissae, which does not take the quantile function as a process to obtain asymptotic theory.  

The present paper proposes a new statistical method for testing such inverse stochastic dominance based on empirical process theory. The test is constructed following the framework of \citet{BDB14} and \citet{Beare2017improved} for testing Lorenz dominance which is highly related to inverse stochastic dominance and stochastic dominance. {See other tests of Lorenz dominance in  \citet{mcfadden1989testing}, 
\citet{bishop1991international}, \citet{bishop1991lorenz}, \citet{dardanoni1999inference}, and \citet{davidson2000statistical}.} Tests on stochastic dominance can be found in \citet{bishop1989asymptotically}, \citet{anderson1996nonparametric}, \citet{davidson2000statistical}, \citet{barrett2003consistent}, and \citet{linton2005consistent}. 
Suppose there are two cumulative distribution functions (CDFs) $F_1:[0,\infty)\rightarrow \mathbb{R}$ and $F_2:[0,\infty)\rightarrow \mathbb{R}$ of two income (or wealth, etc.) distributions in two populations. 
We follow \citet{aaberge2021ranking} and define two functions $\Lambda_1^m$ ($\tilde\Lambda_1^m$) and $\Lambda_2^m$ ($\tilde\Lambda_2^m$) for $F_1$ and $F_2$ that are associated with the $m$th-degree upward (downward) inverse stochastic dominance. The null hypothesis we are interested in is that $\Lambda_1^m\ge\Lambda_2^m$ ($\tilde\Lambda_1^m\ge\tilde\Lambda_2^m$). We then measure the difference between $\Lambda_1^m$ ($\tilde\Lambda_1^m$) and $\Lambda_2^m$ ($\tilde\Lambda_2^m$) by some particular map $\mathcal{F}$. We construct the test statistic based on $\mathcal{F}$ and establish the asymptotic distribution of the test statistic applying the results from \citet{fang2014inference} and \citet{K17}. We then construct the critical value for the test by employing the bootstrap method of \citet{fang2014inference}. The simulation studies demonstrate the good finite sample properties of the test. 
Finally, we apply our test to the empirical example of the inequality growth in the United Kingdom discussed by \citet{aaberge2021ranking}. Our testing results show that both $3$rd-degree upward and downward inverse stochastic dominances provide a relatively complete ranking of the income distributions.

\emph{\textbf{Notation.}}
We follow \citet{li2022Unified} and introduce the following standard notation.
Throughout the paper, we suppose all the random elements are defined on
a probability space $( \Omega, \mathcal{A}, \mathbb{P} )$. 
Let $\ell^\infty(A)$
denote the set of all bounded real-valued functions on $A$ for every arbitrary set $A$. We equip $\ell^\infty(A)$ with the supremum norm $\Vert \cdot \Vert_{\infty}$ such that
$\Vert f \Vert_{\infty}=\sup_{x\in A} \vert f(x) \vert$ for every $f\in \ell^\infty(A)$.
For every subset $A$ of a metric space,
let $C(A)$ be the set of all continuous real-valued functions on $A$.
For $p\ge1$ and $A\subset\mathbb{R}$, let $L^p(A)$ denote the space of measurable functions such that $\int_{A}|f(t)|^p\mathrm{d}t<\infty$ for every $f\in L^p(A)$. Equip $L^p(A)$ with the norm $\Vert\cdot \Vert_{p}$ such that $\Vert f \Vert_{p}=(\int_{A}|f(t)|^p\mathrm{d}t)^{1/p}$ for every $f\in L^p(A)$.
Let $\leadsto$ denote the weak convergence defined in \citet[p.~4]{van1996weak}. Let $\overset{\mathbb{P}}{\leadsto}$ denote the weak convergence in probability conditional on the sample as defined in \citet[p.~19]{kosorok2008introduction}.

\section{Hypotheses for $m$th-degree Inverse Stochastic Dominance}

We follow \citet{Beare2017improved} and suppose that $F_1$ and $F_2$ satisfy the following assumption.
\begin{assumption}\label{ass.distribution}
\citep{Beare2017improved}
	For $j=1,2$, the CDF $F_{j}$ satisfies $F_j(0)=0$ and is continuously differentiable on the interior of its support, with strictly positive derivative. In addition, each $F_{j}$ has finite $(2+\epsilon)$th absolute moment for some $\epsilon>0$.
\end{assumption}

We let $Q_1$ and $Q_2$ denote the quantile functions corresponding to $F_1$ and $F_2$, respectively. By definition, we have
\begin{align}
Q_j(p)&=\inf\left\{x\in [0,\infty):F_j(x)\ge p \right\},\quad p\in[0,1].
\end{align}
If $F_j$ has finite first moment $\mu_j$ (it does under Assumption \ref{ass.distribution}), the quantile function $Q_j$ is integrable with $\int_0^1 Q_j(p)\mathrm{d}p=\mu_j$. In this case, for $m\ge2$, we follow \citet{aaberge2021ranking} and define functions corresponding to $F_j$ by
\begin{align}
{\Lambda}_{j}^{m}\left(  p\right)  &=\int_{0}^{p}\cdots\int_{0}^{t_{3}}\int_{0}^{t_{2}}%
{Q}_{j}\left(  t_{1}\right)  \mathrm{d}t_{1}\mathrm{d}t_{2}\cdots
\mathrm{d}t_{m-1}\notag\\
&=\frac{1}{(m-2)!}\int_{0}^{p}(p-t)^{m-2}Q_j(t)\mathrm{d}t, \quad p\in[0,1],
\end{align}
and 
\begin{align}
{\tilde{\Lambda}}_{j}^{m}\left(  p\right)  &=\int_{p}^{1}\cdots
\int_{t_{3}}^{1}\int_{0}^{t_{2}}{Q}_{j}\left(  t_{1}\right)
\mathrm{d}t_{1}\mathrm{d}t_{2}\cdots\mathrm{d}t_{m-1}\notag\\
&=\frac{1}{(m-2)!}\left[(1-p)^{m-2}\mu_j-\int_{p}^{1}(t-p)^{m-2}Q_j(t)\mathrm{d}t\right], \quad p\in[0,1].
\end{align}
With these functions, for $m\ge3$, we introduce the upward and downward inverse stochastic dominances in \citet{aaberge2021ranking}. 

\begin{definition}
A distribution $F_{1}$ $m$th-degree upward inverse stochastically
dominates a distribution $F_{2}$ if $\Lambda_{1}^{m}\left(  p\right)
\geq\Lambda_{2}^{m}\left(  p\right)  $ for all $p\in\left[  0,1\right]  $.
\end{definition}

\begin{definition}
A distribution $F_{1}$ $m$th-degree downward inverse
stochastically dominates a distribution $F_{2}$ if $\tilde{\Lambda}_{1}%
^{m}\left(  p\right)  \geq\tilde{\Lambda}_{2}^{m}\left(  p\right)  $ for all
$p\in\left[  0,1\right]  $.    
\end{definition}

The definitions above can be viewed as weak versions of those in \citet{aaberge2021ranking} since we do not need the inequalities to hold strictly for some point $p\in (0,1)$.
In this paper, we focus on $m\ge3$. When $m=2$, the upward inverse stochastic dominance is the generalized Lorenz dominance. Propositions 2.3 and 2.4 of \citet{aaberge2021ranking} provide results linking the $m$th-degree upward and downward inverse stochastic dominances to Gini social welfare functions and Lorenz social welfare functions. Also, Theorems 2.3 and 2.4 of \citet{aaberge2021ranking} show the relationship between the $m$th-degree upward and downward inverse stochastic dominances and the general family of welfare functions.  
Under the above types of dominances, the two sets of hypotheses of interest in this paper are
\begin{enumerate}[label=(\roman*)]
    \item Upward inverse stochastic dominance:
    \begin{align*}
\text{H}_{0}  &  :\Lambda_{2}^m(p)  \leq \Lambda_{1}^m(p)  \text{ for
	all }p\in[  0,1]  ,\\
\text{H}_{1}  &  :\Lambda_{2}^m(p)  > \Lambda_{1}^m(p)  \text{ for some
}p\in[  0,1]  .
\end{align*}

    \item Downward inverse stochastic dominance:
    \begin{align*}
\tilde{\mathrm{H}}_{0}  &  :\tilde{\Lambda}_{2}^m(p)  \leq \tilde{\Lambda}_{1}^m(p)  \text{ for
	all }p\in[  0,1]  ,\\
\tilde{\mathrm{H}}_{1}  &  :\tilde{\Lambda}_{2}^m(p)  > \tilde{\Lambda}_{1}^m(p)  \text{ for some
}p\in[  0,1]  .
\end{align*}
\end{enumerate}
Now, we define two difference functions by
\begin{align*}
    \phi_m^u(p)=\Lambda_{2}^m(p)  - \Lambda_{1}^m(p) \text{ and }\phi_m^d(p)=\tilde{\Lambda}_{2}^m(p)  - \tilde{\Lambda}_{1}^m(p), \quad p\in[0,1].
\end{align*}
These functions will be used to construct the tests in the following sections. 

\section{Test Formulation}

Following \citet{BDB14} and \citet{Beare2017improved}, we consider two sampling frameworks from $F_1$ and $F_2$, in either of which, for $j=1,2$, we draw an independently and identically distributed (i.i.d.) sample $\{X_i^j\}_{i=1}^{n_j}$ from $F_j$.

In the first sampling framework ({independent samples}), we suppose that the two samples are independent of each other. The sample sizes $n_1$ and $n_2$ may be different and may be treated as functions of an underlying index $n\in\mathbb N$. We assume that as $n\to\infty$,
\begin{align}\label{samplesizes}
\frac{n_1n_2}{n_1+n_2}\to\infty\quad\text{and}\quad\frac{n_1}{n_1+n_2}\to\lambda\in[0,1].
\end{align}

In the second sampling framework ({matched pairs}), we suppose that $n_1=n_2=n$ and the pairs $\{(X_i^1,X_i^2)\}_{i=1}^n$ are i.i.d. We allow the dependence between paired observations. We use $C$ to denote the bivariate copula function characterizing this dependence, and we suppose that $C$ has maximal correlation strictly less than one \citep[see, e.g.,][Definition 3.2]{beare2010copulas}.  Clearly, $\lambda=1/2$ in this framework.

We summarize these settings in the following assumption. 

\begin{assumption}\label{ass.data}\citep{BDB14,Beare2017improved} 
	The i.i.d.\ samples $\{X_i^1\}_{i=1}^{n_1}$ and $\{X_i^2\}_{i=1}^{n_2}$ satisfy one of the following conditions.
	\begin{enumerate}[label=(\roman*)]
		\item \emph{Independent samples}: The samples $\{X_i^1\}_{i=1}^{n_1}$ and $\{X_i^2\}_{i=1}^{n_2}$ are independent of each other. The sample sizes $n_1$ and $n_2$ satisfy \eqref{samplesizes}.
		\item \emph{Matched pairs}: The sample sizes $n_1$ and $n_2$ satisfy $n_1=n_2=n$ for some index $n$. The sample pairs $\{(X_i^1,X_i^2)\}_{i=1}^{n}$ are i.i.d., and the bivariate copula function $C$ has maximal correlation strictly less than one.
	\end{enumerate}
\end{assumption}

With the random samples, for $j=1,2$, define the empirical CDF
\begin{align*}
\hat{F}_j(x)=\frac{1}{n_j}\sum_{i=1}^{n_j}
{1}(X_i^j\leq x),\quad x\in[0,\infty),
\end{align*}
and the empirical quantile function
\begin{align*}
\hat{Q}_j(p)=\inf\{x\in[0,\infty):\hat{F}_j(x)\geq p\},\quad p\in[0,1].
\end{align*}
Following \citet{Beare2017improved}, we let $\mathcal{B}$ be a centered Gaussian random element in $C(\left[0,1\right]  ^{2})$ with the covariance kernel
\[
Cov\left(  \mathcal{B}\left(  u,v\right)  ,\mathcal{B(}u^{\prime},v^{\prime
})\right)  =C\left(  u\wedge u^{\prime},v\wedge v^{\prime}\right)  -C\left(
u,v\right)  C(u^{\prime},v^{\prime}),
\]
where $u,u',v,v'\in[0,1]$, and $C$ is the copula function in Assumption \ref{ass.data}. Under Assumption \ref{ass.data}(i), $C(u,v)=uv$. Under Assumption \ref{ass.data}(ii), $C$ is the unique copula function for the pair $(X_{i}^{1},X_{i}^{2})  $. Define $\mathcal{B}_{1}$ and $\mathcal{B}_{2}$ to be
the centered Gaussian random elements in $C(\left[  0,1\right]  )$ such that
$\mathcal{B}_{1}\left(  u\right)  =\mathcal{B}\left(  u,1\right)  $ and
$\mathcal{B}_{2}\left(  v\right)  =\mathcal{B}\left(  1,v\right)  $. As mentioned by \citet{Beare2017improved}, it is
straightforward to show that 
\[
\left(
\begin{array}
[c]{c}%
n_{1}^{1/2}\left(  \hat{F}_{1}-F_{1}\right)  \\
n_{2}^{1/2}\left(  \hat{F}_{2}-F_{2}\right)
\end{array}
\right)  \leadsto\left(
\begin{array}
[c]{c}%
\mathcal{B}_{1}\circ F_{1}\\
\mathcal{B}_{2}\circ F_{2}%
\end{array}
\right)
\]
in $\ell^{\infty}\left(  [0,\infty)\right)  \times\ell^{\infty}\left(
[0,\infty)\right)  $.
\citet{Beare2017improved} show that under Assumptions \ref{ass.distribution} and \ref{ass.data}, by applying the results of \citet{K17}, we have
\[
\left(
\begin{array}
[c]{c}%
n_{1}^{1/2}\left(  \hat{Q}_{1}-Q_{1}\right)  \\
n_{2}^{1/2}\left(  \hat{Q}_{2}-Q_{2}\right)
\end{array}
\right)  \leadsto\left(
\begin{array}
[c]{c}%
-Q_{1}^{\prime}\cdot\mathcal{B}_{1}\\
-Q_{2}^{\prime}\cdot\mathcal{B}_{2}%
\end{array}
\right)
\]
in $L^{1}\left(  [0,1]\right)  \times L^{1}\left(  [0,1]\right)  $. Recall that for all $p\in\left[  0,1\right]  $, $\Lambda_{j}^{2}\left(
p\right)  =\int_{0}^{p}Q_{j}\left(  t\right)  \mathrm{d}t$ and we estimate
$\Lambda_{j}^{2}\left(  p\right)  $ by $\hat{\Lambda}_{j}^{2}\left(  p\right)
=\int_{0}^{p}\hat{Q}_{j}\left(  t\right)  \mathrm{d}t$. Then by continuous
mapping theorem,
\[
\left(
\begin{array}
[c]{c}%
n_{1}^{1/2}\left(  \hat{\Lambda}_{1}^2-{\Lambda}_{1}^2\right)  \\
n_{2}^{1/2}\left(  \hat{\Lambda}_{2}^2-{\Lambda}_{2}^2\right)
\end{array}
\right)
\leadsto\binom{\mathcal{V}_{1}}{\mathcal{V}_{2}},
\]
where $\mathcal{V}_{j}(p)=\int_{0}^{p}-Q_{j}^{\prime}\left(  t\right)
\mathcal{B}_{j}(t)\mathrm{d}t$. By continuous mapping theorem again, it follows
that
\[
\sqrt{T_{n}}\left\{  (\hat{\Lambda}_{2}^{2}-\hat{\Lambda}_{1}^{2})-\left(
\Lambda_{2}^{2}-\Lambda_{1}^{2}\right)  \right\}  \leadsto\mathbb{G}_{\Lambda},
\]
where $T_n=n_1n_2/(n_1+n_2)$, $\mathbb{G}_{\Lambda}=\sqrt{\lambda
}\mathcal{V}_{2}-\sqrt{1-\lambda}\mathcal{V}_{1}$, and for all $p,p^{\prime}\in\left[  0,1\right]  $,
\begin{align*}
E\left[  \mathbb{G}_{\Lambda}\left(  p\right)  \mathbb{G}_{\Lambda}(
p^{\prime})  \right]    =&\,\left(  1-\lambda\right)  E\left[
\mathcal{V}_{1}\left(  p\right)  \mathcal{V}_{1}(  p^{\prime})
\right]  -\sqrt{\lambda\left(  1-\lambda\right)  }E\left[  \mathcal{V}%
_{1}\left(  p\right)  \mathcal{V}_{2}(  p^{\prime})  \right]  \\
& -\sqrt{\lambda\left(  1-\lambda\right)  }E\left[  \mathcal{V}_{2}\left(
p\right)  \mathcal{V}_{1}(  p^{\prime})  \right]  +\lambda E\left[
\mathcal{V}_{2}\left(  p\right)  \mathcal{V}_{2}(  p^{\prime})
\right]  .
\end{align*}

The following lemma provides a formula for computing $E[  \mathcal{V}_{j}\left(  p\right)  \mathcal{V}_{j^{\prime}}(
p^{\prime})  ]$ which is useful when we estimate $E[  \mathbb{G}_{\Lambda}\left(  p\right)  \mathbb{G}_{\Lambda}(
p^{\prime})  ]$.
\begin{lemma}\label{lemma.EVV}
For $j,j^{\prime}\in\left\{  1,2\right\}  $ and $p,p^{\prime}\in\left[
0,1\right]  $,
\[
E\left[  \mathcal{V}_{j}\left(  p\right)  \mathcal{V}_{j^{\prime}}(
p^{\prime})  \right]  =Cov\left(  Q_{j}\left(  p\right)  \wedge
X^{j},Q_{j^{\prime}}(  p^{\prime})  \wedge X^{j^{\prime}}\right)  .
\]    
\end{lemma}

For $j=1,2$ and $m\geq3$, we define the estimators for the functions $\Lambda_j^m$ and $\tilde\Lambda_j^m$ by
\begin{align*}
\hat{\Lambda}_{j}^{m}\left(  p\right)  =\int_{0}^{p}\cdots\int_{0}^{t_{3}}\int_{0}^{t_{2}}%
\hat{Q}_{j}\left(  t_{1}\right)  \mathrm{d}t_{1}\mathrm{d}t_{2}\cdots
\mathrm{d}t_{m-1} =\int_{0}^{p}\cdots\int_{0}^{t_{3}}\hat{\Lambda}_{j}^{2}\left(  t_{2}\right)  \mathrm{d}t_{2}\cdots
\mathrm{d}t_{m-1}
\end{align*}
and 
\begin{align*}
\hat{\tilde{\Lambda}}_{j}^{m}\left(  p\right)  =\int_{p}^{1}\cdots
\int_{t_{3}}^{1}\int_{0}^{t_{2}}\hat{Q}_{j}\left(  t_{1}\right)
\mathrm{d}t_{1}\mathrm{d}t_{2}\cdots\mathrm{d}t_{m-1}=\int_{p}^{1}\cdots\int_{t_{3}}^{1}\hat{\Lambda}_{j}^{2}\left(  t_{2}\right)
\mathrm{d}t_{2}\cdots\mathrm{d}t_{m-1}.
\end{align*}
Define the empirical difference functions
\[
\hat{\phi}_m^u\left(  p\right)  =\hat{\Lambda}_{2}^{m}\left(  p\right) -\hat{\Lambda}_{1}^{m}\left(  p\right) \text{ and } {\hat{\phi}}_m^d\left(  p\right)  =\hat{\tilde{\Lambda}}_{2}^{m}\left(  p\right) -\hat{\tilde{\Lambda}}_{1}^{m}\left(  p\right), \quad p\in[0,1].
\]
Then under Assumptions \ref{ass.distribution} and \ref{ass.data}, by continuous mapping theorem, we have that 
\begin{align}\label{eq.weak convergence phi^u}
    \sqrt{T_n}(\hat{\phi}_m^u-{\phi}_m^u)\leadsto \mathbb{G}_m^u \text{ with } \mathbb{G}_m^u(p)=
\int_{0}^{p}\cdots\int_{0}^{t_{3}}\mathbb{G}_{\Lambda}\left(  t_{2}\right)  \mathrm{d}t_{2}\cdots \mathrm{d}t_{m-1}
\end{align}
and 
\begin{align}\label{eq.weak convergence phi^d}
    \sqrt{T_n}(\hat{\phi}_m^d-{\phi}_m^d)\leadsto \mathbb{G}_m^d\text{ with } \mathbb{G}_m^d(p)=\int_{p}^{1}\cdots\int_{t_{3}}^{1}\mathbb{G}_{\Lambda} \left(  t_{2}\right)\mathrm{d}t_{2}\cdots\mathrm{d}t_{m-1}.
\end{align}
By Fubini's Theorem, we have that
\begin{align*}
Var(\mathbb{G}_m^u(p))=\int_{0}^{p}\cdots\int_{0}^{t_{3}^{\prime}}\left(  \int_{0}^{p}\cdots\int_{0}^{t_{3}}E\left[
\mathbb{G}_{\Lambda}(  t_{2})  \mathbb{G}_{\Lambda}(  t_{2}^{\prime
})  \right]  \mathrm{d}t_{2}\cdots\mathrm{d}t_{m-1}\right)  \mathrm{d}t_{2}^{\prime}\cdots\mathrm{d}t_{m-1}^{\prime}
\end{align*}
and 
\begin{align*}
Var(\mathbb{G}_m^d(p))=    \int_{p}^{1}\cdots\int_{t_{3}^{\prime}}^1\left(  \int_{p}^{1}\cdots\int_{t_{3}}^1E\left[
\mathbb{G}_{\Lambda}(  t_{2})  \mathbb{G}_{\Lambda}(  t_{2}^{\prime
})  \right]  \mathrm{d}t_{2}\cdots\mathrm{d}t_{m-1}\right)  \mathrm{d}t_{2}^{\prime}\cdots\mathrm{d}t_{m-1}^{\prime}
\end{align*}
for every $p\in[0,1]$.

We set the test statistics as $T_n^{1/2}\mathcal F(\hat{\phi}_m^w)$ for $w\in\{u,d\}$, where $\mathcal F:C([0,1])\to\mathbb R$ is some map that measures the size of the positive part of $\hat{\phi}_m^w$. The test statistics we consider are similar to those of \citet{BDB14} and \citet{Beare2017improved}. We assume that the map $\mathcal{F}$ possesses the following properties \citep{BDB14}.
\begin{assumption}
	\label{ass.functional}\citep{BDB14}
 The map $\mathcal{F}:C([0,1])\to\mathbb R$ satisfies that, for all $h\in C([0,1])$,
	\begin{enumerate}[label=(\roman*)]
		\item if $h(p)\le 0$ for all $p\in[0,1]$ and $h(p)=0$ for some $p\in[0,1]$, then $\mathcal{F}(h)=0$;
		\item if $h(p)>0$ for some $p\in(0,1)$, then $\mathcal{F}(h)>0$. 
	\end{enumerate}
\end{assumption}
By definition, we have that $\phi_m^u(0)=\phi_m^d(1)=0$. Therefore, under Assumption
\ref{ass.functional}, the null hypothesis ${\mathrm {H}}_{0}$ ($\tilde{\mathrm{H}}_{0}$) is true if and only if $\mathcal F(\phi_m^u)=0$ ($\mathcal F(\phi_m^d)=0$), while the alternative hypothesis $\mathrm H_1$ ($\tilde{\mathrm{H}}_{1}$) is true if and only if $\mathcal{F}(\phi_m^u)>0$ ($\mathcal{F}(\phi_m^d)>0$). 

Following \citet{BDB14} and \citet{Beare2017improved}, we focus on two choices of $\mathcal F$ denoted by $\mathcal S$ and $\mathcal I$: For every $h\in C([0,1])$, 
\begin{align}\label{defSI}
\mathcal{S}(h)  =\sup_{p\in[  0,1]
}h(  p)  \quad \text{and}\quad\mathcal{I}(h)  =\int_{0}^{1}\max\{h(  p),0\}\mathrm{d}p.
\end{align}
It can easily be verified that both $\mathcal S$ and $\mathcal I$ satisfy Assumption \ref{ass.functional} and both of them are Hadamard directionally differentiable as illustrated in the following. We exploit this property to obtain the asymptotics of the test statistics. Following \citet{fang2014inference}, we introduce the definition of Hadamard directional differentiability.\footnote{See more discussions and examples on Hadamard directional differentiability in \citet{shapiro1991asymptotic}, \citet{dumbgen1993nondifferentiable}, \citet{andrews2000inconsistency}, \citet{bickel2012resampling}, \citet{hirano2012impossibility}, \citet{beare2015nonparametric},
	\citet{Beare2016global}, \citet{hansen2017regression}, \citet{Seo2016tests},  \citet{Beare2015improved}, \citet{chen2019improved}, \citet{Beare2017improved}, and \citet{sun2018ivvalidity}.} 

\begin{definition}
	\label{def.Hadamard directional}
 \citep{fang2014inference}
 Let $\mathbb{D}$ and $\mathbb{E}$ be normed spaces. A map $\mathcal{G}:\mathbb{D}\to\mathbb{E}$ is said to be Hadamard
	directionally differentiable at $\phi\in\mathbb{D}$ tangentially to $\mathbb{D}_0\subset\mathbb{D}$ if there is a
continuous	map $\mathcal{G}_{\phi}^{\prime}:\mathbb{D}_0\rightarrow
	\mathbb{E}$ such that
	\begin{align}\label{Hadamard directional derivative}
		\lim_{n\rightarrow\infty}\left\Vert\frac{\mathcal{G}(\phi+t_{n}h_{n}%
			)-\mathcal{G}(\phi)}{t_{n}}-\mathcal{G}_{\phi}^{\prime}(h)\right\Vert
		_{\mathbb{E}}=0
	\end{align}
	for all sequences $\{h_{n}\}\subset\mathbb{D}$ and $\{t_{n}\}\subset
	\mathbb{R}_{+}$ such that $t_{n}\downarrow0$ and $h_{n}\rightarrow h\in\mathbb{D}_0$.
\end{definition}

We suppose that the map $\mathcal F$ is Hadamard directionally differentiable in the next assumption. 
\begin{assumption}\label{ass.Hadamard directional differentiability}
    The map $\mathcal{F}:C([0,1])\to \mathbb{R}$ is Hadamard directionally differentiable at $\phi_m^u$ and $\phi_m^d$ with the directional derivatives $\mathcal{F}'_{\phi_m^u}$ and $\mathcal{F}'_{\phi_m^d}$, respectively. 
\end{assumption}

For every $\phi\in C([0,1])$, we define the set
	\begin{align*}
	\Psi(\phi)=\argmax_{p\in[0,1]}\phi(p).
    \end{align*}
By Lemma S.4.9 of \citet{fang2014inference}, we have that the map $\mathcal S$ satisfies Assumption \ref{ass.Hadamard directional differentiability}, and for $w\in\{u,d\}$ and $m\ge 3$,
	\begin{align*}
	\mathcal S'_{\phi_m^w}(h)=\sup_{p\in\Psi(\phi_m^w)}h(p),\quad h\in C([0,1]).
	\end{align*}

For every $\phi\in C([0,1])$, define the sets
	\begin{align*}
	B_{0}(\phi)=\{p\in[0,1]:\phi(p)=0\}\quad\text{and}\quad B_{+}(\phi)=\{p\in[0,1]:\phi(p)>0\}.
	\end{align*}
By Lemma S.4.5 of \citet{fang2014inference}, the map $\mathcal I$ satisfies Assumption \ref{ass.Hadamard directional differentiability}, and for $w\in\{u,d\}$ and $m\ge 3$,
	\begin{align*}
	\mathcal I'_{\phi_m^w}(h)=\int_{B_{+}(  \phi_m^w)
	}h(  p)   \mathrm{d}p+\int_{B_{0}(
	\phi_m^w)  }\max\{  h(  p)  ,0\}    \mathrm{d}p,\quad h\in C([0,1]).
\end{align*}

The following lemma provides the asymptotic distributions of the test statistics by applying a more general version of the delta method for Hadamard directionally differentiable maps \citep{shapiro1991asymptotic,dumbgen1993nondifferentiable,fang2014inference}.
\begin{lemma}\label{lemma.weak convergence F}
	Under Assumptions \ref{ass.distribution}, \ref{ass.data}, and \ref{ass.Hadamard directional differentiability},  it follows that for $m\ge 3$ and $w\in\{u,d\}$,
	\begin{align*}
	T_n^{1/2}(\mathcal F(\hat\phi_m^w)-\mathcal F(\phi_m^w))\rightsquigarrow\mathcal F'_{\phi_m^w}(\mathbb{G}_m^w).
	\end{align*}
	Suppose, in addition, that Assumption \ref{ass.functional} holds. If $\mathrm{H}_0$ and $\tilde{\mathrm{H}}_0$ are true, then
	\begin{align}\label{eq.F asymptotic limit}
	T_n^{1/2}\mathcal F(\hat\phi_m^w)\rightsquigarrow\mathcal F'_{\phi_m^w}(\mathbb{G}_m^w),
	\end{align}
	whereas if $\mathrm{H}_1$ and $\tilde{\mathrm{H}}_1$ are true then $T_n^{1/2}\mathcal F(\hat\phi_m^w)\to\infty$ in probability.
\end{lemma}

\subsection{Bootstrap Approximation}
We construct the bootstrap approximations following the method of  \citet{BDB14} and \citet{Beare2017improved}.

First, using random weights (specified below) $W^1_{n_1}=(W^1_{1,n_1},\ldots,W^1_{n_1,n_1})$ and $W^2_{n_2}=(W^2_{1,n_2},\ldots,W^2_{n_2,n_2})$, we construct the bootstrap versions of $\hat F_1$ and $\hat F_2$ by
\begin{align*}
\hat F^\ast_j(x)=\frac{1}{n_j}\sum_{i=1}^{n_j}W^j_{i,n_j}{1}(X_i^j\leq x),\quad x\in[0,\infty).
\end{align*}
For independent samples, the random weights $W^1_{n_1}=(W^1_{1,n_1},\ldots,W^1_{n_1,n_1})$ and $W^2_{n_2}=(W^2_{1,n_2},\ldots,W^2_{n_2,n_2})$ are drawn independently of the data and of one another from the multinomial distributions $M(n_1,p_{11},\ldots,p_{1n_1})$ and $M(n_2,p_{21},\ldots,p_{2n_2})$, respectively, where the probabilities over the categories $\{1,\ldots,n_1\}$ and $\{1,\ldots,n_2\}$ satisfy $p_{11}=\cdots=p_{1n_1}=1/n_1$ and $p_{21}=\cdots=p_{2n_2}=1/n_2$.
For matched pairs, we set $W^1_n=W^2_n$, and draw this vector independently of the data from the multinomial distribution with equal probabilities over the categories $\{1,\ldots,n\}$. 


With $\hat{F}_j^*$, we then define bootstrap empirical quantile functions
\begin{align*}
\hat Q^\ast_j(p)=\inf\{x\in[0,\infty):\hat F^\ast_j(x)\geq p\},\quad p\in[0,1].
\end{align*}
For $j=1,2$, $m\geq3$, and $p\in[0,1]$, define
\begin{align*}
\hat{\Lambda}_{j}^{m*}\left(  p\right)  =\int_{0}^{p}\cdots\int_{0}^{t_{3}}\int_{0}^{t_{2}}%
\hat{Q}_{j}^*\left(  t_{1}\right)  \mathrm{d}t_{1}\mathrm{d}t_{2}\cdots
\mathrm{d}t_{m-1}
\end{align*}
and 
\begin{align*}
\hat{\tilde{\Lambda}}_{j}^{m*}\left(  p\right)  =\int_{p}^{1}\cdots
\int_{t_{3}}^{1}\int_{0}^{t_{2}}\hat{Q}_{j}^*\left(  t_{1}\right)
\mathrm{d}t_{1}\mathrm{d}t_{2}\cdots\mathrm{d}t_{m-1}.
\end{align*}
Define the bootstrap difference functions
\[
\hat{\phi}_m^{u*}\left(  p\right)  =\hat{\Lambda}_{2}^{m*}\left(  p\right) -\hat{\Lambda}_{1}^{m*}\left(  p\right) \text{ and } {\hat{\phi}}_m^{d*}\left(  p\right)  =\hat{\tilde{\Lambda}}_{2}^{m*}\left(  p\right) -\hat{\tilde{\Lambda}}_{1}^{m*}\left(  p\right), \quad p\in[0,1].
\]

The directional derivative $\mathcal F'_{\phi_m^w}$ depends on the underlying DGP and therefore is unknown. We provide some estimator $\widehat{\mathcal F'_{\phi_m^w}}$ to consistently estimate $\mathcal F'_{\phi_m^w}$ and use it to construct the critical value. In the following, we provide estimators $\widehat{\mathcal S'_{\phi_m^w}}$ and $\widehat{\mathcal I'_{\phi_m^w}}$ for the derivatives of $\mathcal S$ and $\mathcal I$, respectively. 

Let $\hat{\lambda}=n_1/(n_1+n_2)$. For $j,j^{\prime}\in\left\{  1,2\right\}  $ and $p,p^{\prime}\in\left[
0,1\right]  $, we estimate $E[  \mathcal{V}_{j}\left(  p\right)
\mathcal{V}_{j^{\prime}}\left(  p^{\prime}\right)  ]  $ by
\[
\hat{E}\left[  \mathcal{V}_{j}\left(  p\right)  \mathcal{V}_{j^{\prime}%
}(p^{\prime})\right]  =\widehat{Cov}\left(  Q_{j}\left(  p\right)  \wedge
X^{j},Q_{j^{\prime}}(p^{\prime})\wedge X^{j^{\prime}}\right)  ,
\]
where $\widehat{Cov}(  Q_{j}\left(  p\right)  \wedge X^{j},Q_{j^{\prime}%
}\left(  p^{\prime}\right)  \wedge X^{j^{\prime}})  $ is the sample
covariance of the two samples
\begin{align*}
\left\{\hat{Q}_j(p)\wedge X^j_i\right\}_{i=1}^{n_j} \text{ and } \left\{\hat{Q}_{j'}(p')\wedge X^{j'}_i\right\}_{i=1}^{n_{j'}}.
\end{align*}
For independent samples, we estimate $E\left[  \mathbb{G}_{\Lambda}\left(
p\right)  \mathbb{G}_{\Lambda}\left(  p^{\prime}\right)  \right]  $ by
\[
\hat{E}\left[  \mathbb{G}_{\Lambda}\left(  p\right)  \mathbb{G}_{\Lambda
}(p^{\prime})\right]  =(  1-\hat{\lambda})  \hat{E}\left[
\mathcal{V}_{1}\left(  p\right)  \mathcal{V}_{1}(p^{\prime})\right]
+\hat{\lambda}\hat{E}\left[  \mathcal{V}_{2}\left(  p\right)  \mathcal{V}%
_{2}(p^{\prime})\right]  .
\]
For matched pairs,
\begin{align*}
\hat{E}\left[  \mathbb{G}_{\Lambda}\left(  p\right)  \mathbb{G}_{\Lambda
}(p^{\prime})\right]    =&\,(  1-\hat{\lambda})  \hat{E}\left[
\mathcal{V}_{1}\left(  p\right)  \mathcal{V}_{1}(p^{\prime})\right]
-\sqrt{\hat{\lambda}(  1-\hat{\lambda})  }\hat{E}\left[
\mathcal{V}_{1}\left(  p\right)  \mathcal{V}_{2}(p^{\prime})\right]  \\
& -\sqrt{\hat{\lambda}(  1-\hat{\lambda})  }\hat{E}\left[
\mathcal{V}_{2}\left(  p\right)  \mathcal{V}_{1}(p^{\prime})\right]
+\hat{\lambda}\hat{E}\left[  \mathcal{V}_{2}\left(  p\right)  \mathcal{V}%
_{2}(p^{\prime})\right]  .
\end{align*}
We then estimate the variances $Var(\mathbb{G}_m^u(p))$ and $Var(\mathbb{G}_m^d(p))$ by
\begin{align*}
    \hat{\sigma}_m^{u2}(p)=\int_{0}^{p}\cdots\int_{0}^{t_{3}^{\prime}}\left(  \int_{0}^{p}\cdots\int_{0}^{t_{3}}\hat{E}\left[
\mathbb{G}_{\Lambda}(  t_{2})  \mathbb{G}_{\Lambda}(  t_{2}^{\prime
})  \right]  \mathrm{d}t_{2}\cdots\mathrm{d}t_{m-1}\right)  \mathrm{d}t_{2}^{\prime}\cdots\mathrm{d}t_{m-1}^{\prime}
\end{align*}
and 
\begin{align*}
    \hat{\sigma}_m^{d2}(p)=\int_{p}^{1}\cdots\int_{t_{3}^{\prime}}^1\left(  \int_{p}^{1}\cdots\int_{t_{3}}^1\hat{E}\left[
\mathbb{G}_{\Lambda}(  t_{2})  \mathbb{G}_{\Lambda}(  t_{2}^{\prime
})  \right]  \mathrm{d}t_{2}\cdots\mathrm{d}t_{m-1}\right)  \mathrm{d}t_{2}^{\prime}\cdots\mathrm{d}t_{m-1}^{\prime},
\end{align*}
respectively, for $p\in[0,1]$.

For $w\in\{u,d\}$, let $\hat{v}_m^w=\max\{\hat{\sigma}_m^{w2}, \xi\}^{1/2}$ for some small positive number $\xi$. We suggest using $\xi=0.001$ in practice. Here, the trimming parameter $\xi$ bounds the estimator $\hat{\sigma}_m^{w2}$ away from zero. Similar bounded estimators are used by \citet{Beare2015improved}, \citet{Beare2017improved}, and \citet{sun2018ivvalidity} to estimate contact sets in different contexts. Define the estimated contact set
\begin{align}\label{estcs}
\widehat{B_0(\phi_m^w)}=\left\{  p\in[  0,1]  :\big| T_n^{1/2}\hat{\phi}_m^w(  p)
\big| \leq\tau_{n}\hat{v}_m^w(p)\right\},
\end{align}
where $\tau_n$ is some tuning parameter. Here, we are using pointwise confidence intervals to estimate the contact set as in \citet{Beare2017improved}: Each point $p\in[0,1]$ is included in the estimated contact set if $T_n^{1/2}\hat{\phi}_m^w(p)$ is less than or equal to $\tau_n$ estimated standard deviations from zero.
We then estimate $\mathcal S'_{\phi_m^w}$ and $\mathcal I'_{\phi_m^w}$ by
\begin{align*}
\widehat{\mathcal S'_{\phi_m^w}}(h)=\sup_{p\in\widehat{B_0(\phi_m^w)}}h(p) \text{ and } \widehat{\mathcal I'_{\phi_m^w}}(h)=\int_{\widehat{B_{0}(
		\phi_m^w)}  }\max\{  h(  p)  ,0\}    \mathrm{d}p,\quad h\in C([0,1]).
\end{align*}
We use $\widehat{B_0(\phi_m^w)}$ to estimate both of the sets $B_0(\phi_m^w)$ and $\Psi(\phi_m^w)$. By definition, the curves $\Lambda_1^m$ and $\Lambda_2^m$ ($\tilde\Lambda_1^m$ and $\tilde\Lambda_2^m$) always touch at zero (one). If the null hypothesis is true then the maximum value of $\phi_m^u$ ($\phi_m^d$) is zero, and thus $\Psi(\phi_m^w)=B_0(\phi_m^w)$. The set $B_+(\phi_m^w)$ is always empty under the null, so we do not estimate it.

\subsection{Asymptotic Properties of the Test}\label{secbsasym}

The following lemma provides the asymptotic limit of the bootstrap process $T_n^{1/2}(\hat\phi_m^{w\ast}-\hat\phi_m^w)$ conditional on the data which consistently approximates the asymptotic limit of $T_n^{1/2}(\hat\phi_m^w-\phi_m^w)$. The result follows from the delta method for the bootstrap in \citet{K17}.
\begin{lemma}\label{lemma.bootstrap phi asymptotic distribution}
	Under Assumptions \ref{ass.distribution} and \ref{ass.data}, we have
	\begin{equation}
	T_{n}^{1/2}(  \hat{\phi}_m^{w\ast}-\hat\phi_m^w)\overset{\mathbb
 {P}}{\leadsto}\mathbb{G}_m^w\text{ in }C([0,1]).
	\end{equation}
\end{lemma}
To ensure that the conditional law of $\widehat{\mathcal F'_{\phi_m^w}}(T_n^{1/2}(\hat\phi_m^{w\ast}-\hat\phi_m^w))$ consistently approximates the distribution of $\mathcal F'_{\phi_m^w}(\mathbb{G}_m^w)$ with the convergence in Lemma \ref{lemma.bootstrap phi asymptotic distribution}, we introduce the following assumption from \citet{fang2014inference} on $\widehat{\mathcal F'_{\phi_m^w}}$.
\begin{assumption}\label{ass.FS}\citep{fang2014inference}
	The estimated map $\widehat{\mathcal F'_{\phi_m^w}}:C([0,1])\to\mathbb R$ satisfies that for every compact $K\subset C([0,1])$ and every $\epsilon>0$,
	\begin{align*}
	\mathbb{P}\left(\sup_{h\in K}\vert\widehat{\mathcal F'_{\phi_m^w}}(h)-\mathcal F'_{\phi_m^w}(h)\vert>\epsilon\right)\to0.
	\end{align*}
\end{assumption}
The following lemma shows that the proposed estimators $\widehat{\mathcal S'_{\phi_m^w}}$ and $\widehat{\mathcal I'_{\phi_m^w}}$ both satisfy Assumption \ref{ass.FS} based on the discussion in \citet{fang2014inference}.
\begin{lemma}\label{lemma.SI}
	Suppose that Assumptions \ref{ass.distribution} and \ref{ass.data} hold with $\tau_n\to\infty$ and $T_n^{-1/2}\tau_n\to0$ as $n\to\infty$. Then the estimators $\widehat{\mathcal S'_{\phi_m^w}}$ and $\widehat{\mathcal I'_{\phi_m^w}}$ satisfy Assumption \ref{ass.FS}.
\end{lemma}
By Theorem 3.2 of \citet{fang2014inference}, the next lemma shows that the distribution of the bootstrap statistic $\widehat{\mathcal F'_{\phi_m^w}}(T_n^{1/2}(\hat\phi_m^{w\ast}-\hat\phi_m^w))$ conditional on the data consistently approximates the distribution of the asymptotic limit $\mathcal F'_{\phi_m^w}(\mathbb{G}_m^w)$.
\begin{lemma}\label{lemma.bsconsistent}
	Under Assumptions \ref{ass.distribution}, \ref{ass.data}, \ref{ass.Hadamard directional differentiability}, and \ref{ass.FS}, it follows that
	\begin{equation}
	\widehat{\mathcal F'_{\phi_m^w}}(T_n^{1/2}(\hat\phi_m^{w\ast}-\hat\phi_m^w))\overset{\mathbb{P}}{\leadsto} \mathcal F'_{\phi_m^w}(\mathbb{G}_m^w)\text{ in }\mathbb R.
	\end{equation}
\end{lemma}

Let $\hat{c}_{m1-\alpha}^w$ denote the $(1-\alpha)$ quantile of the bootstrap law of $\widehat{\mathcal F'_{\phi_m^w}}(T_n^{1/2}(\hat\phi_m^{w\ast}-\hat\phi_m^w))$ conditional on the data:
\begin{align}\label{cvdef}
\hat{c}_{m1-\alpha}^w=\inf\left\{c\in\mathbb R: \mathbb{P}\left(   \widehat{\mathcal F'_{\phi_m^w}}(T_n^{1/2}(\hat\phi_m^{w\ast}-\hat\phi_m^w))\le c\,\middle|\,   \{  X_{i}^{1}\}  _{i=1}^{n_{1}}  ,\{  X_{i}^{2}\}  _{i=1}^{n_{2}}      \right) \ge 1-\alpha   \right\}.
\end{align} 
In practice, we approximate $\hat c_{m1-\alpha}^w$ using the $(1-\alpha)$ quantile of the $B$ independently generated bootstrap statistics, where $B$ is sufficiently large. We set the decision rule for the test as
\begin{align}\label{eq.decision rule}
\text{Reject } \mathrm H_{0} \;(\tilde{\mathrm{H}}_{0}) \text{ if } T_{n}^{1/2}\mathcal{F}(  \hat{\phi}_m^w)>\hat{c}_{m1-\alpha}^w.
\end{align}
The following proposition provides the asymptotic properties of the test.
\begin{proposition}\label{prop.test}
	Suppose that Assumptions \ref{ass.distribution}, \ref{ass.data}, \ref{ass.functional}, \ref{ass.Hadamard directional differentiability}, and \ref{ass.FS} are satisfied.
	\begin{enumerate}[label=(\roman*)]
		\item If $\mathrm H_{0}$ ($\tilde{\mathrm{H}}_{0}$) is true, and the CDF of $\mathcal F'_{\phi_m^w}(\mathbb{G}_m^w)$ is continuous and strictly increasing at its $1-\alpha$ quantile, then $\mathbb P(T_{n}^{1/2}\mathcal{F}(  \hat{\phi}_m^w)>\hat{c}_{m1-\alpha}^w)\to\alpha$.
		\item If $\mathrm H_{0}$ ($\tilde{\mathrm{H}}_{0}$) is false, then for both $\mathcal{F}=\mathcal{S}$ and $\mathcal{F}=\mathcal{I}$, $\mathbb P(T_{n}^{1/2}\mathcal{F}(  \hat{\phi}_m^w)>\hat{c}_{m1-\alpha}^w)\to1$.
	\end{enumerate}
\end{proposition}

Proposition \ref{prop.test} shows that the test is asymptotically size controlled and consistent. By Theorem 11.1 of \citet{davydov1998local}, for $\mathcal{F}=\mathcal{S}$ and $\mathcal{F}=\mathcal{I}$, if the asymptotic limit $\mathcal{F}'_{\phi_m^w}(\mathbb{G}_m^w)\neq0$, then the CDF of $\mathcal{F}'_{\phi_m^w}(\mathbb{G}_m^w)$ is differentiable and has a positive derivative everywhere except at countably many points in its support.
As discussed in \citet{Beare2017improved}, the null configurations for which $\mathcal{F}'_{\phi_m^w}(\mathbb{G}_m^w)\neq0$ constitute the boundary of the null as defined in \citet{LSW10}. The null configurations that are not on the boundary may include the cases where $\Lambda_1^m=\Lambda_2^m$ ($\tilde\Lambda_1^m=\tilde\Lambda_2^m$) happens only at $0$ ($1$) or in a set with measure $0$. If $\mathcal{F}'_{\phi_m^w}(\mathbb{G}_m^w)=0$ under the null, the test statistic converges to zero in probability by Lemma \ref{lemma.weak convergence F} and so does the bootstrap critical value by Lemma \ref{lemma.bsconsistent}. Proposition \ref{prop.test} does not provide a result on how the rejection rate would behave in this case. 
As mentioned by \citet{Beare2017improved} and \citet{sun2018ivvalidity}, this is a common theoretical limitation for irregular testing problems. In practice, we may replace the bootstrap critical value $\hat{c}^w_{m1-\alpha}$ with $\max\{\hat{c}^w_{m1-\alpha},\eta\}$ or $\hat{c}^w_{m1-\alpha}+\eta$, where $\eta$ is some extremely small positive number \citep[p.~13]{DH16}. Monte Carlo simulations in Section \ref{sec.simulation} show that the rejection rates of our test are well controlled at null configurations with $\eta=0$.

\section{Simulation Evidence}\label{sec.simulation}


In this section, we show the finite sample properties of the proposed test via a set of Monte Carlo simulations. All the simulations consider independent samples. The experimental replications are set to $1000$. The warp-speed method of \citet{giacomini2013warp} is employed to expedite the simulations. The nominal significance level $\alpha=0.05$.
To estimate the contact set in the test statistic, we use five different
tuning parameter values: $\tau_n=1,2,3,4,\infty$.
In the simulations for the size control of the test, we let $n_1 = n_2 = 2000$.  We choose the value of the tuning parameter based on the empirical size of the test. In the simulations for the empirical power of the test, we set $n_1 = n_2=n \in \{200,500,1000,2000\}$ to show the consistency of the test.

We design the data generating processes (DGPs) based on income distributions belonging to the double Pareto parametric
family following \citet{Beare2017improved}.\footnote{\citet{reed2001pareto,reed2003pareto} and \citet{toda2012double} show that income distributions can be well approximated by members of the double Pareto parametric
family.} 
For some $M>0$, the probability density function of double Pareto distribution is as follows: 
\[
f\left(  x\right)  =\left\{
\begin{array}
[c]{c}%
\frac{\alpha\beta}{\alpha+\beta}M^{\alpha}x^{-\alpha-1}\\
\frac{\alpha\beta}{\alpha+\beta}M^{-\beta}x^{\beta-1}%
\end{array}
\right.
\begin{array}
[c]{c}%
x\geq M,\\
0\leq x<M.
\end{array}
\]
The scale parameter $M$ is set to one, and the
shape parameters $\alpha,\beta>0$ are set to different values in our simulations. We follow \citet{Beare2017improved} and write
$X\sim\mathrm{dP}(\alpha,\beta)$ to indicate that a random variable $X$ has the
double Pareto distribution with $M=1$ and shape
parameters $\alpha, \beta$. As mentioned in \citet{Beare2017improved}, when $\alpha>2$, the distribution $\mathrm{dP}(\alpha,\beta)$ satisfies Assumption \ref{ass.distribution}. We set $m=3$ in all simulations.

\subsection{Size Control and Tuning Parameter Selection}
We first study the empirical size of the test and choose the value of the tuning parameter $\tau_n$ in finite samples. We generate i.i.d.\ data for $X^1$ and $X^2$ from the same distribution $\mathrm{dP}(\alpha,\beta)$ and report the rejection rates for $\alpha\in \{2,3,4,5\}$\footnote{When $\alpha=2$, Assumption \ref{ass.distribution} is violated. Our results show that the size of the test is also controlled in this case.} and $\beta\in\{1,\ldots,8\}$. The rejection rates are reported in Tables \ref{tab:Rej H0 ISD up} and \ref{tab:Rej H0 ISD down}. The simulation results show that when $\tau_n=3$, the rejection rates are close to those in the conservative case where $\tau_n=\infty$ and are close to $\alpha$. Based on these results, we suggest using $\tau_n=3$ for sample sizes $n_1\le 2000$ and $n_2\le 2000$. When the sample sizes increase, $\tau_n$ may be increased accordingly. 
	
	\begin{table}[h]
		
		\centering
		\caption{Rejection Rates under $\mathrm H_{0}$ for $3$rd-degree Upward Inverse Stochastic Dominance}
		\scalebox{0.9}{
			\begin{tabular}{  c  c  c  c  c  c  c  c  c  c  c  }
				\hline
				\hline
				\multirow{2}{*}{$\mathcal{F}$} & \multirow{2}{*}{$\alpha$} & \multirow{2}{*}{$\tau_n$}& \multicolumn{8}{c}{$\beta$} \\
				\cline{4-11}
				& & & 1 & 2 & 3 & 4 & 5 & 6 & 7 & 8  \\
				\hline
				\multirow{20}{*}{$\mathcal{F}=\mathcal{S}$}	
				&   & 1 & 0.071 & 0.078 & 0.095 & 0.083 & 0.078 & 0.052 & 0.061 & 0.097 \\ 
				&   & 2 & 0.052 & 0.060 & 0.074 & 0.067 & 0.053 & 0.038 & 0.049 & 0.079 \\ 
				& 2 & 3 & 0.051 & 0.060 & 0.065 & 0.057 & 0.053 & 0.038 & 0.049 & 0.079 \\ 
				&   & 4 & 0.051 & 0.060 & 0.065 & 0.056 & 0.052 & 0.038 & 0.049 & 0.079 \\ 
				&   & $\infty$ & 0.051 & 0.060 & 0.065 & 0.056 & 0.052 & 0.038 & 0.049 & 0.079 \\ 
				\cline{2-11}
				&   & 1 & 0.120 & 0.067 & 0.065 & 0.088 & 0.062 & 0.067 & 0.059 & 0.070 \\ 
				&   & 2 & 0.070 & 0.052 & 0.054 & 0.056 & 0.049 & 0.053 & 0.047 & 0.058 \\ 
				& 3 & 3 & 0.070 & 0.043 & 0.050 & 0.056 & 0.047 & 0.053 & 0.047 & 0.058 \\ 
				&   & 4 & 0.070 & 0.043 & 0.050 & 0.056 & 0.043 & 0.053 & 0.047 & 0.058 \\ 
				&   & $\infty$ & 0.070 & 0.043 & 0.050 & 0.056 & 0.043 & 0.053 & 0.047 & 0.058 \\ 
				\cline{2-11}
				&   & 1 & 0.065 & 0.062 & 0.069 & 0.067 & 0.097 & 0.050 & 0.071 & 0.075 \\ 
				&   & 2 & 0.053 & 0.050 & 0.053 & 0.050 & 0.066 & 0.037 & 0.057 & 0.052 \\ 
				& 4 & 3 & 0.051 & 0.046 & 0.051 & 0.048 & 0.059 & 0.037 & 0.057 & 0.052 \\ 
				&   & 4 & 0.051 & 0.046 & 0.051 & 0.048 & 0.059 & 0.037 & 0.057 & 0.052 \\ 
				&   & $\infty$ & 0.051 & 0.046 & 0.051 & 0.048 & 0.059 & 0.037 & 0.057 & 0.052 \\ 
				\cline{2-11}
				&   & 1 & 0.102 & 0.075 & 0.062 & 0.053 & 0.073 & 0.073 & 0.059 & 0.071 \\ 
				&   & 2 & 0.080 & 0.063 & 0.043 & 0.050 & 0.050 & 0.064 & 0.042 & 0.057 \\ 
				& 5 & 3 & 0.078 & 0.061 & 0.043 & 0.045 & 0.048 & 0.064 & 0.041 & 0.052 \\ 
				&   & 4 & 0.078 & 0.061 & 0.043 & 0.045 & 0.048 & 0.064 & 0.041 & 0.051 \\ 
				&   & $\infty$ & 0.078 & 0.061 & 0.043 & 0.045 & 0.048 & 0.064 & 0.041 & 0.051 \\ 
				\hline
				\multirow{20}{*}{$\mathcal{F}=\mathcal{I}$}	
				&   & 1 & 0.076 & 0.085 & 0.112 & 0.078 & 0.087 & 0.076 & 0.076 & 0.113 \\ 
				&   & 2 & 0.048 & 0.061 & 0.075 & 0.060 & 0.052 & 0.045 & 0.051 & 0.068 \\ 
				& 2 & 3 & 0.046 & 0.052 & 0.063 & 0.052 & 0.052 & 0.045 & 0.047 & 0.067 \\ 
				&   & 4 & 0.046 & 0.052 & 0.063 & 0.052 & 0.052 & 0.045 & 0.047 & 0.067 \\ 
				&   & $\infty$ & 0.046 & 0.052 & 0.063 & 0.052 & 0.052 & 0.045 & 0.047 & 0.067 \\ 
				\cline{2-11}
				&   & 1 & 0.110 & 0.068 & 0.077 & 0.072 & 0.071 & 0.078 & 0.082 & 0.077 \\ 
				&   & 2 & 0.058 & 0.053 & 0.046 & 0.057 & 0.047 & 0.053 & 0.049 & 0.054 \\ 
				& 3 & 3 & 0.057 & 0.048 & 0.045 & 0.053 & 0.046 & 0.053 & 0.046 & 0.051 \\ 
				&   & 4 & 0.057 & 0.048 & 0.045 & 0.053 & 0.046 & 0.053 & 0.045 & 0.051 \\ 
				&   & $\infty$ & 0.057 & 0.048 & 0.045 & 0.053 & 0.046 & 0.053 & 0.045 & 0.051 \\
				\cline{2-11}
				&   & 1 & 0.088 & 0.076 & 0.082 & 0.064 & 0.072 & 0.059 & 0.086 & 0.100 \\ 
				&   & 2 & 0.052 & 0.050 & 0.057 & 0.047 & 0.050 & 0.032 & 0.069 & 0.048 \\ 
				& 4 & 3 & 0.051 & 0.049 & 0.054 & 0.041 & 0.043 & 0.032 & 0.069 & 0.048 \\ 
				&   & 4 & 0.051 & 0.049 & 0.054 & 0.041 & 0.043 & 0.030 & 0.069 & 0.048 \\ 
				&   & $\infty$ & 0.051 & 0.049 & 0.054 & 0.041 & 0.043 & 0.030 & 0.069 & 0.048 \\ 
				\cline{2-11}
				&   & 1 & 0.101 & 0.084 & 0.072 & 0.053 & 0.084 & 0.083 & 0.064 & 0.065 \\ 
				&   & 2 & 0.075 & 0.069 & 0.040 & 0.042 & 0.047 & 0.068 & 0.041 & 0.052 \\ 
				& 5 & 3 & 0.073 & 0.068 & 0.039 & 0.035 & 0.047 & 0.066 & 0.040 & 0.049 \\ 
				&   & 4 & 0.073 & 0.068 & 0.039 & 0.035 & 0.047 & 0.066 & 0.040 & 0.048 \\ 
				&   & $\infty$ & 0.073 & 0.068 & 0.039 & 0.035 & 0.047 & 0.066 & 0.040 & 0.048 \\ 
				
				\hline
				\hline

			\end{tabular}
		}
		
		\label{tab:Rej H0 ISD up}
	\end{table}

	\begin{table}[h]
		
		\centering
		\caption{Rejection Rates under $\tilde{\mathrm H}_{0}$ for $3$rd-degree Downward Inverse Stochastic Dominance}
		\scalebox{0.9}{
			\begin{tabular}{  c  c  c  c  c  c  c  c  c  c  c  }
				\hline
				\hline
				\multirow{2}{*}{$\mathcal{F}$} & \multirow{2}{*}{$\alpha$} & \multirow{2}{*}{$\tau_n$}& \multicolumn{8}{c}{$\beta$} \\
				\cline{4-11}
				& & & 1 & 2 & 3 & 4 & 5 & 6 & 7 & 8  \\
				\hline
				\multirow{20}{*}{$\mathcal{F}=\mathcal{S}$}	
				&   & 1 & 0.070 & 0.073 & 0.091 & 0.083 & 0.077 & 0.049 & 0.059 & 0.096 \\ 
				&   & 2 & 0.052 & 0.069 & 0.074 & 0.067 & 0.054 & 0.038 & 0.049 & 0.078 \\ 
				& 2 & 3 & 0.051 & 0.060 & 0.065 & 0.057 & 0.053 & 0.038 & 0.049 & 0.078 \\ 
				&   & 4 & 0.051 & 0.060 & 0.065 & 0.055 & 0.051 & 0.038 & 0.049 & 0.078 \\ 
				&   & $\infty$ & 0.051 & 0.060 & 0.065 & 0.055 & 0.051 & 0.038 & 0.049 & 0.078 \\ 
				\cline{2-11}
				&   & 1 & 0.120 & 0.067 & 0.066 & 0.082 & 0.055 & 0.067 & 0.062 & 0.072 \\ 
				&   & 2 & 0.070 & 0.052 & 0.053 & 0.056 & 0.049 & 0.053 & 0.047 & 0.058 \\ 
				& 3 & 3 & 0.070 & 0.043 & 0.050 & 0.056 & 0.047 & 0.053 & 0.047 & 0.058 \\ 
				&   & 4 & 0.070 & 0.043 & 0.050 & 0.056 & 0.043 & 0.053 & 0.047 & 0.058 \\ 
				&   & $\infty$ & 0.070 & 0.043 & 0.050 & 0.056 & 0.043 & 0.053 & 0.047 & 0.058 \\ 
				\cline{2-11}
				&   & 1 & 0.065 & 0.056 & 0.060 & 0.068 & 0.097 & 0.051 & 0.070 & 0.075 \\ 
				&   & 2 & 0.052 & 0.050 & 0.052 & 0.050 & 0.065 & 0.038 & 0.056 & 0.053 \\ 
				& 4 & 3 & 0.050 & 0.046 & 0.052 & 0.048 & 0.059 & 0.037 & 0.056 & 0.052 \\ 
				&   & 4 & 0.050 & 0.046 & 0.052 & 0.048 & 0.059 & 0.037 & 0.056 & 0.052 \\ 
				&   & $\infty$ & 0.050 & 0.046 & 0.052 & 0.048 & 0.059 & 0.037 & 0.056 & 0.052 \\ 
				\cline{2-11}
				&   & 1 & 0.096 & 0.071 & 0.062 & 0.054 & 0.074 & 0.080 & 0.058 & 0.073 \\ 
				&   & 2 & 0.080 & 0.063 & 0.043 & 0.050 & 0.050 & 0.064 & 0.042 & 0.056 \\ 
				& 5 & 3 & 0.071 & 0.061 & 0.043 & 0.045 & 0.049 & 0.064 & 0.042 & 0.051 \\ 
				&   & 4 & 0.071 & 0.061 & 0.043 & 0.045 & 0.049 & 0.064 & 0.042 & 0.051 \\ 
				&   & $\infty$ & 0.071 & 0.061 & 0.043 & 0.045 & 0.049 & 0.064 & 0.042 & 0.051 \\ 
				\hline
				\multirow{20}{*}{$\mathcal{F}=\mathcal{I}$}	
				&   & 1 & 0.076 & 0.080 & 0.106 & 0.087 & 0.080 & 0.053 & 0.058 & 0.098 \\ 
				&   & 2 & 0.054 & 0.062 & 0.076 & 0.070 & 0.054 & 0.044 & 0.050 & 0.077 \\ 
				& 2 & 3 & 0.049 & 0.060 & 0.072 & 0.058 & 0.050 & 0.044 & 0.048 & 0.074 \\ 
				&   & 4 & 0.049 & 0.060 & 0.072 & 0.057 & 0.048 & 0.044 & 0.048 & 0.074 \\ 
				&   & $\infty$ & 0.049 & 0.060 & 0.072 & 0.057 & 0.048 & 0.044 & 0.048 & 0.074 \\ 
				\cline{2-11}
				&   & 1 & 0.127 & 0.078 & 0.072 & 0.092 & 0.077 & 0.073 & 0.066 & 0.092 \\ 
				&   & 2 & 0.079 & 0.046 & 0.051 & 0.052 & 0.051 & 0.049 & 0.046 & 0.059 \\ 
				& 3 & 3 & 0.077 & 0.044 & 0.050 & 0.051 & 0.049 & 0.049 & 0.044 & 0.059 \\ 
				&   & 4 & 0.077 & 0.044 & 0.050 & 0.051 & 0.049 & 0.049 & 0.044 & 0.059 \\ 
				&   & $\infty$ & 0.077 & 0.044 & 0.050 & 0.051 & 0.049 & 0.049 & 0.044 & 0.059 \\
				\cline{2-11} 
				&   & 1 & 0.075 & 0.059 & 0.085 & 0.070 & 0.102 & 0.053 & 0.081 & 0.094 \\ 
				&   & 2 & 0.057 & 0.051 & 0.055 & 0.049 & 0.067 & 0.038 & 0.053 & 0.046 \\ 
				& 4 & 3 & 0.051 & 0.047 & 0.052 & 0.048 & 0.063 & 0.037 & 0.052 & 0.046 \\ 
				&   & 4 & 0.051 & 0.047 & 0.051 & 0.048 & 0.063 & 0.037 & 0.052 & 0.046 \\ 
				&   & $\infty$ & 0.051 & 0.047 & 0.051 & 0.048 & 0.063 & 0.037 & 0.052 & 0.046 \\ 
				\cline{2-11}
				&   & 1 & 0.110 & 0.080 & 0.067 & 0.064 & 0.078 & 0.082 & 0.066 & 0.082 \\ 
				&   & 2 & 0.084 & 0.055 & 0.040 & 0.051 & 0.060 & 0.071 & 0.050 & 0.056 \\ 
				& 5 & 3 & 0.075 & 0.053 & 0.040 & 0.048 & 0.055 & 0.069 & 0.049 & 0.053 \\ 
				&   & 4 & 0.075 & 0.053 & 0.040 & 0.048 & 0.055 & 0.069 & 0.049 & 0.052 \\ 
				&   & $\infty$ & 0.075 & 0.053 & 0.040 & 0.048 & 0.055 & 0.069 & 0.049 & 0.052 \\ 
				
				\hline
				\hline

			\end{tabular}
		}
		
		\label{tab:Rej H0 ISD down}
	\end{table}
	
	\subsection{Empirical Power}
	
	We now study the empirical power of the test in finite samples using $\tau_n=3$ as selected above.  	
	In the simulations for upward inverse stochastic dominance, we generate i.i.d.\ data $\{X_i^1\}_{i=1}^{n_1}$ for $X^1$ from $\mathrm{dP}(2.1,1.5)$, and i.i.d.\ data $\{X_i^2\}_{i=1}^{n_2}$ for $X^2$ from $\mathrm{dP}(100,\beta)$ whose law depends on $\beta$. In this setting, $\mathrm{H}_0$ does not hold with $m=3$. We let $\beta$ vary between $2.91$ and $3$ in increments of $0.01$. The rejection rates are reported in Table \ref{tab:Rej H1 ISD up}.
	In the simulations for downward inverse stochastic dominance, we generate i.i.d.\ data $\{X_i^1\}_{i=1}^{n_1}$ for $X^1$ from $\mathrm{dP}(2.1,1.5)$, and i.i.d.\ data $\{X_i^2\}_{i=1}^{n_2}$ for $X^2$ from $\mathrm{dP}(\alpha,4)$ whose law depends on $\alpha$. We let $\alpha$ vary between $10$ and $100$
	in increments of $10$.  The rejection rates are reported in Table \ref{tab:Rej H1 ISD down}.

    All the results show that as the sample size increases, the empirical power increases to $1$, which demonstrates the good finite sample power property of the test.

		\begin{table}[h]
		
		\centering
		\caption{Rejection Rates under $\mathrm H_{1}$ for $3$rd-degree Upward Inverse Stochastic Dominance}
		\scalebox{0.9}{
			\begin{tabular}{  c  c  c  c  c  c  c  c  c  c  c  c   }
				\hline
				\hline
				\multirow{2}{*}{$\mathcal{F}$} &  \multirow{2}{*}{$n$}& \multicolumn{10}{c}{$\beta$} \\
				\cline{3-12}
				& & 2.91 & 2.92 & 2.93 & 2.94 & 2.95 & 2.96 & 2.97 & 2.98 & 2.99 & 3 \\
				\hline
				\multirow{4}{*}{$\mathcal{F}=\mathcal{S}$}	
				& 200 & 0.058 & 0.054 & 0.052 & 0.064 & 0.068 & 0.077 & 0.067 & 0.082 & 0.085 & 0.065 \\ 
				& 500 & 0.189 & 0.194 & 0.256 & 0.217 & 0.296 & 0.225 & 0.251 & 0.318 & 0.184 & 0.274 \\ 
				& 1000 & 0.634 & 0.560 & 0.436 & 0.576 & 0.670 & 0.655 & 0.682 & 0.679 & 0.650 & 0.755 \\ 
				& 2000 & 0.974 & 0.981 & 0.976 & 0.952 & 0.974 & 0.981 & 0.969 & 0.989 & 0.993 & 0.998 \\  
				\hline
				\multirow{4}{*}{$\mathcal{F}=\mathcal{I}$}	
				& 200 & 0.310 & 0.373 & 0.354 & 0.387 & 0.372 & 0.336 & 0.431 & 0.412 & 0.391 & 0.435 \\ 
				& 500 & 0.791 & 0.746 & 0.810 & 0.799 & 0.852 & 0.855 & 0.853 & 0.897 & 0.833 & 0.878 \\ 
				& 1000 & 0.996 & 0.993 & 0.995 & 0.993 & 0.998 & 0.997 & 0.998 & 0.993 & 0.998 & 0.997 \\ 
				& 2000 & 1.000 & 1.000 & 1.000 & 1.000 & 1.000 & 1.000 & 1.000 & 1.000 & 1.000 & 1.000 \\    

				\hline
				\hline		
				
			\end{tabular}
		}
		
		\label{tab:Rej H1 ISD up}
	\end{table}

	\begin{table}[h]
		
		\centering
		\caption{Rejection Rates under $\tilde{\mathrm H}_{1}$ for $3$rd-degree Downward Inverse Stochastic Dominance}
		\scalebox{0.9}{
			\begin{tabular}{  c  c  c  c  c  c  c  c  c  c  c c }
				\hline
				\hline
				\multirow{2}{*}{$\mathcal{F}$} &  \multirow{2}{*}{$n$}& \multicolumn{10}{c}{$\alpha$} \\
				\cline{3-12}
				& & 10 & 20 & 30 & 40 & 50 & 60 & 70 & 80 & 90 & 100 \\
				\hline
				\multirow{4}{*}{$\mathcal{F}=\mathcal{S}$}	
				& 200 & 0.818 & 0.580 & 0.458 & 0.412 & 0.360 & 0.327 & 0.336 & 0.338 & 0.317 & 0.305 \\ 
				& 500 & 0.996 & 0.899 & 0.828 & 0.747 & 0.722 & 0.645 & 0.628 & 0.626 & 0.567 & 0.569 \\ 
				& 1000 & 1.000 & 0.994 & 0.963 & 0.928 & 0.918 & 0.891 & 0.884 & 0.844 & 0.812 & 0.858 \\ 
				& 2000 & 1.000 & 1.000 & 0.999 & 0.993 & 0.995 & 0.996 & 0.981 & 0.987 & 0.981 & 0.989 \\  
				\hline
				\multirow{4}{*}{$\mathcal{F}=\mathcal{I}$}	
				& 200 & 0.519 & 0.226 & 0.134 & 0.126 & 0.089 & 0.094 & 0.096 & 0.087 & 0.067 & 0.066 \\ 
				& 500 & 0.971 & 0.661 & 0.414 & 0.297 & 0.290 & 0.213 & 0.185 & 0.217 & 0.155 & 0.178 \\ 
				& 1000 & 1.000 & 0.957 & 0.735 & 0.665 & 0.576 & 0.498 & 0.452 & 0.422 & 0.381 & 0.372 \\ 
				& 2000 & 1.000 & 1.000 & 0.997 & 0.976 & 0.950 & 0.927 & 0.884 & 0.852 & 0.851 & 0.842 \\  
				
				\hline
				\hline		
				
			\end{tabular}
		}
		
		\label{tab:Rej H1 ISD down}
	\end{table}

\section{Empirical Application}

We revisit an empirical example of the inequality growth in the United Kingdom discussed by \citet{aaberge2021ranking} to show the performance of the proposed test in practice. As illustrated by \citet{aaberge2021ranking}, the data come from the European Community Household Panel (ECHP) for 1995–2001, and from the European Union Statistics on Income and Living Conditions (EU-SILC) for 2005–2010. See more details about the data in \citet[p.~657]{aaberge2021ranking}.  \citet{aaberge2021ranking} find that the use of 3rd-degree upward dominance works little for raising the ability to rank income distributions, while 3rd-degree downward dominance provides an almost complete ranking of the income distributions.
	
	We use the same data to reconduct the tests. 
 For two years A and B, we test two null hypotheses: (i) Year A dominates Year B and (ii) Year B dominates Year A. 
		If we reject (i), while not rejecting (ii), we conclude that B strictly dominates A, that is $\Lambda_B^{m}(p)\ge \Lambda_A^{m}(p)$ ($\tilde\Lambda_B^{m}(p)\ge \tilde\Lambda_A^{m}(p)$) for all $p\in[0,1]$ and the inequality holds strictly for some $p\in[0,1]$, denoted by A $<$ B or B $>$ A. Similarly, if we cannot reject (i), 
		while rejecting (ii), we conclude that A strictly dominates B. Otherwise, we conclude
		that we do not find a strict dominance relationship between A and B.
 Tables \ref{tab:Application_S_upwards}--\ref{tab:Application_I_downwards} show the test results for the $3$rd-degree upward and downward dominances obtained from our test using $\mathcal{F}=\mathcal{S}$ and $\mathcal{F}=\mathcal{I}$. From these results, we conclude that both the upward and downward dominance tests can provide a relatively complete ranking of the income distributions. 
	
	\begin{table}[h]
		
		\centering
		\caption{Ranking of Income
			Distributions by $3$rd-degree
			Upward Dominance Using $\mathcal{S}$}
		\scalebox{0.8}{
			\begin{tabular}{  c  c  c  c  c  c  c  c  c  c  c  c  c  c }
				\hline
				\hline

				Year & 1995 & 1996 & 1997 & 1998 & 1999 & 2000 & 2001 & 2005 & 2006 & 2007 & 2008 & 2009 & 2010 \\ \hline
				1994 & ~ & ~ & $<$ & $<$ & $<$ & $<$ & $<$ & $<$ & $<$ & $<$ & $<$ & $<$ & $<$ \\ 
				1995 & ~ & ~ & $<$ & $<$ & $<$ & $<$ & $<$ & $<$ & $<$ & $<$ & $<$ & $<$ & $<$ \\ 
				1996 & ~ & ~ & $<$ & $<$ & $<$ & $<$ & $<$ & $<$ & $<$ & $<$ & $<$ & $<$ & $<$ \\ 
				1997 & ~ & ~ & ~ & ~ & ~ & $<$ & $<$ & $<$ & $<$ & $<$ & $<$ & $<$ & $<$ \\ 
				1998 & ~ & ~ & ~ & ~ & ~ & $<$ & $<$ & $<$ & $<$ & $<$ & $<$ & $<$ & $<$ \\ 
				1999 & ~ & ~ & ~ & ~ & ~ & $<$ & $<$ & $<$ & $<$ & $<$ & $<$ & $<$ & $<$ \\ 
				2000 & ~ & ~ & ~ & ~ & ~ & ~ & ~ & $<$ & $<$ & $<$ & $<$ & $>$ & ~ \\ 
				2001 & ~ & ~ & ~ & ~ & ~ & ~ & ~ & ~ & ~ & ~ & ~ & $>$ & $>$ \\ 
				2005 & ~ & ~ & ~ & ~ & ~ & ~ & ~ & ~ & ~ & $<$ & $>$ & $>$ & $>$ \\ 
				2006 & ~ & ~ & ~ & ~ & ~ & ~ & ~ & ~ & ~ & ~ & $>$ & $>$ & $>$ \\ 
				2007 & ~ & ~ & ~ & ~ & ~ & ~ & ~ & ~ & ~ & ~ & $>$ & $>$ & $>$ \\ 
				2008 & ~ & ~ & ~ & ~ & ~ & ~ & ~ & ~ & ~ & ~ & ~ & $>$ & $>$ \\ 
				2009 & ~ & ~ & ~ & ~ & ~ & ~ & ~ & ~ & ~ & ~ & ~ & ~ & ~ \\

				\hline
				\hline		
				
			\end{tabular}
		}
		
		\label{tab:Application_S_upwards}
	\end{table}

	\begin{table}[h]
		
		\centering
		\caption{Ranking of Income
			Distributions by $3$rd-degree
			Downward Dominance Using $\mathcal{S}$}
		\scalebox{0.8}{
			\begin{tabular}{  c  c  c  c  c  c  c  c  c  c  c  c  c  c }
				\hline
				\hline
				
				Year & 1995 & 1996 & 1997 & 1998 & 1999 & 2000 & 2001 & 2005 & 2006 & 2007 & 2008 & 2009 & 2010 \\ \hline
				1994 & ~ & ~ & $<$ & $<$ & $<$ & $<$ & $<$ & $<$ & $<$ & $<$ & $<$ & $<$ & $<$ \\ 
				1995 & ~ & ~ & $<$ & $<$ & $<$ & $<$ & $<$ & $<$ & $<$ & $<$ & $<$ & $<$ & $<$ \\ 
				1996 & ~ & ~ & $<$ & $<$ & $<$ & $<$ & $<$ & $<$ & $<$ & $<$ & $<$ & $<$ & $<$ \\ 
				1997 & ~ & ~ & ~ & ~ & ~ & $<$ & $<$ & $<$ & $<$ & $<$ & $<$ & $<$ & $<$ \\ 
				1998 & ~ & ~ & ~ & ~ & ~ & $<$ & $<$ & $<$ & $<$ & $<$ & $<$ & $<$ & $<$ \\ 
				1999 & ~ & ~ & ~ & ~ & ~ & $<$ & $<$ & $<$ & $<$ & $<$ & $<$ & $<$ & $<$ \\ 
				2000 & ~ & ~ & ~ & ~ & ~ & ~ & ~ & $<$ & $<$ & $<$ & $<$ & $>$ & ~ \\ 
				2001 & ~ & ~ & ~ & ~ & ~ & ~ & ~ & $<$ & $<$ & $<$ & $<$ & $>$ & $>$ \\ 
				2005 & ~ & ~ & ~ & ~ & ~ & ~ & ~ & ~ & ~ & ~ & $>$ & $>$ & $>$ \\ 
				2006 & ~ & ~ & ~ & ~ & ~ & ~ & ~ & ~ & ~ & ~ & $>$ & $>$ & $>$ \\ 
				2007 & ~ & ~ & ~ & ~ & ~ & ~ & ~ & ~ & ~ & ~ & $>$ & $>$ & $>$ \\ 
				2008 & ~ & ~ & ~ & ~ & ~ & ~ & ~ & ~ & ~ & ~ & ~ & $>$ & $>$ \\ 
				2009 & ~ & ~ & ~ & ~ & ~ & ~ & ~ & ~ & ~ & ~ & ~ & ~ & ~ \\ 
				
				\hline
				\hline		
				
			\end{tabular}
		}
		
		\label{tab:Application_S_downwards}
	\end{table}

	\begin{table}[h]
		
		\centering
		\caption{Ranking of Income
			Distributions by $3$rd-degree
			Upward Dominance Using $\mathcal{I}$}
		\scalebox{0.8}{
			\begin{tabular}{  c  c  c  c  c  c  c  c  c  c  c  c  c  c }
				\hline
				\hline
				
				Year & 1995 & 1996 & 1997 & 1998 & 1999 & 2000 & 2001 & 2005 & 2006 & 2007 & 2008 & 2009 & 2010 \\ \hline
				1994 & ~ & ~ & $<$ & $<$ & $<$ & $<$ & $<$ & $<$ & $<$ & $<$ & $<$ & $<$ & $<$ \\ 
				1995 & ~ & ~ & $<$ & $<$ & $<$ & $<$ & $<$ & $<$ & $<$ & $<$ & $<$ & $<$ & $<$ \\ 
				1996 & ~ & ~ & $<$ & $<$ & $<$ & $<$ & $<$ & $<$ & $<$ & $<$ & $<$ & $<$ & $<$ \\ 
				1997 & ~ & ~ & ~ & ~ & ~ & $<$ & $<$ & $<$ & $<$ & $<$ & $<$ & $<$ & $<$ \\ 
				1998 & ~ & ~ & ~ & ~ & ~ & $<$ & $<$ & $<$ & $<$ & $<$ & $<$ & $<$ & $<$ \\ 
				1999 & ~ & ~ & ~ & ~ & ~ & $<$ & $<$ & $<$ & $<$ & $<$ & $<$ & $<$ & $<$ \\ 
				2000 & ~ & ~ & ~ & ~ & ~ & ~ & ~ & $<$ & $<$ & $<$ & $<$ & $>$ & $>$ \\ 
				2001 & ~ & ~ & ~ & ~ & ~ & ~ & ~ & $<$ & $<$ & $<$ & $<$ & $>$ & $>$ \\
				2005 & ~ & ~ & ~ & ~ & ~ & ~ & ~ & ~ & ~ & $<$ & $>$ & $>$ & $>$ \\ 
				2006 & ~ & ~ & ~ & ~ & ~ & ~ & ~ & ~ & ~ & ~ & $>$ & $>$ & $>$ \\ 
				2007 & ~ & ~ & ~ & ~ & ~ & ~ & ~ & ~ & ~ & ~ & $>$ & $>$ & $>$ \\ 
				2008 & ~ & ~ & ~ & ~ & ~ & ~ & ~ & ~ & ~ & ~ & ~ & $>$ & $>$ \\ 
				2009 & ~ & ~ & ~ & ~ & ~ & ~ & ~ & ~ & ~ & ~ & ~ & ~ & ~ \\ 
				
				\hline
				\hline		
				
			\end{tabular}
		}
		
		\label{tab:Application_I_upwards}
	\end{table}

		\begin{table}[h]
		
		\centering
		\caption{Ranking of Income
			Distributions by $3$rd-degree
			Downward Dominance Using $\mathcal{I}$}
		\scalebox{0.8}{
			\begin{tabular}{  c  c  c  c  c  c  c  c  c  c  c  c  c  c }
				\hline
				\hline
				
				Year & 1995 & 1996 & 1997 & 1998 & 1999 & 2000 & 2001 & 2005 & 2006 & 2007 & 2008 & 2009 & 2010 \\ \hline
				1994 & ~ & ~ & $<$ & $<$ & $<$ & $<$ & $<$ & $<$ & $<$ & $<$ & $<$ & $<$ & $<$ \\ 
				1995 & ~ & ~ & $<$ & $<$ & $<$ & $<$ & $<$ & $<$ & $<$ & $<$ & $<$ & $<$ & $<$ \\ 
				1996 & ~ & ~ & $<$ & $<$ & $<$ & $<$ & $<$ & $<$ & $<$ & $<$ & $<$ & $<$ & $<$ \\ 
				1997 & ~ & ~ & ~ & $<$ & $<$ & $<$ & $<$ & $<$ & $<$ & $<$ & $<$ & $<$ & $<$ \\ 
				1998 & ~ & ~ & ~ & ~ & ~ & $<$ & $<$ & $<$ & $<$ & $<$ & $<$ & $<$ & $<$ \\ 
				1999 & ~ & ~ & ~ & ~ & ~ & $<$ & $<$ & $<$ & $<$ & $<$ & $<$ & $<$ & $<$ \\  
				2000 & ~ & ~ & ~ & ~ & ~ & ~ & ~ & $<$ & $<$ & $<$ & $<$ & ~ & ~ \\ 
				2001 & ~ & ~ & ~ & ~ & ~ & ~ & ~ & $<$ & $<$ & $<$ & $<$ & $>$ & $>$ \\ 
				2005 & ~ & ~ & ~ & ~ & ~ & ~ & ~ & ~ & ~ & ~ & $>$ & $>$ & $>$ \\ 
				2006 & ~ & ~ & ~ & ~ & ~ & ~ & ~ & ~ & ~ & ~ & $>$ & $>$ & $>$ \\ 
				2007 & ~ & ~ & ~ & ~ & ~ & ~ & ~ & ~ & ~ & ~ & $>$ & $>$ & $>$ \\ 
				2008 & ~ & ~ & ~ & ~ & ~ & ~ & ~ & ~ & ~ & ~ & ~ & $>$ & $>$ \\ 
				2009 & ~ & ~ & ~ & ~ & ~ & ~ & ~ & ~ & ~ & ~ & ~ & ~ & ~ \\ 
				
				\hline
				\hline		
				
			\end{tabular}
		}
		
		\label{tab:Application_I_downwards}
	\end{table}

\section{Proofs of Results}
\begin{proof}[Proof of Lemma \ref{lemma.EVV}]
We closely follow the proof of Lemma A.2 in \citet{Beare2017improved}.
For $j=1,2$, $p\in\left[  0,1\right]  $, and $t\in\left[  0,1\right]
$, define%
\[
h_{j,p}\left(  t\right)  =-1\left(  t\leq p\right)  Q_{j}^{\prime}\left(
t\right)  .
\]
The almost sure integrability of $h_{j,p}\mathcal{B}_{j}$ follows
from the weak convergence of $n_{j}^{1/2}(\hat{Q}_{j}-Q_{j})\leadsto-Q_{j}^{\prime
}\mathcal{B}_{j}$ in $L^{1}\left(  \left[  0,1\right]  \right)  $ established by \citet{K17}. Define
\[
H_{j,p}\left(  u\right)  =\int_{0}^{u}h_{j,p}\left(  t\right)  \mathrm{d}%
t\text{ and }\mathcal{V}_{j}\left(  p\right)  =\int_{0}^{1}h_{j,p}\left(
t\right)  \mathcal{B}_{j}\left(  t\right)  \mathrm{d}t,\quad u,p\in\left[
0,1\right]  .
\]
For $j,j^{\prime}\in\left\{  1,2\right\}  $ and $p,p^{\prime}\in\left[
0,1\right]  $, by Fubini's theorem,
\[
Cov\left(  \mathcal{V}_{j}\left(  p\right)  ,\mathcal{V}_{j^{\prime}}(
p^{\prime})\right)    =\int_{0}^{1}\int_{0}^{1}h_{j,p}\left(
t_{1}\right)  h_{j^{\prime},p^{\prime}}\left(  t_{2}\right)  E\left[
\mathcal{B}_{j}\left(  t_{1}\right)  \mathcal{B}_{j^{\prime}}\left(
t_{2}\right)  \right]  \mathrm{d}t_{1}\mathrm{d}t_{2}.
\]

Let $\left(  U,V\right)  $ be a pair of random variables with joint CDF given
by the copula $C$. For $j=1$ and $j^{\prime}=2$, $E\left[  \mathcal{B}%
_{1}\left(  t_{1}\right)  \mathcal{B}_{2}\left(  t_{2}\right)  \right]
=C\left(  t_{1},t_{2}\right)  -t_{1}t_{2}$. By Theorem 3.1 of \citet{lo2017functional} (see also \citet{cuadras2002covariance} and \citet{beare2009generalization}), we have that 
\[
Cov\left(  H_{j,p}\left(  U\right)  ,H_{j^{\prime},p^{\prime}}\left(
V\right)  \right)  =\int_{0}^{1}\int_{0}^{1}h_{j,p}\left(  t_{1}\right)
h_{j^{\prime},p^{\prime}}\left(  t_{2}\right)  \left(  C\left(  t_{1}%
,t_{2}\right)  -t_{1}t_{2}\right)  \mathrm{d}t_{1}\mathrm{d}t_{2}.
\]
Then it follows that
\begin{align*}
E\left[  \mathcal{V}_{1}\left(  p\right)  \mathcal{V}_{2}(  p^{\prime
})  \right]   &  =Cov\left(  \mathcal{V}_{1}\left(  p\right)
,\mathcal{V}_{2}(  p^{\prime})  \right)  =Cov\left(  H_{1,p}\left(
F_{1}\left(  X^{1}\right)  \right)  ,H_{2,p^{\prime}}\left(  F_{2}\left(
X^{2}\right)  \right)  \right)  \\
&  =Cov\left(  Q_{1}\left(  p\right)  \wedge X^{1},Q_{2}(  p^{\prime
})  \wedge X^{2}\right)  .
\end{align*}
For $j=j^{\prime}$, $E\left[  \mathcal{B}_{j}\left(  t_{1}\right)
\mathcal{B}_{j}\left(  t_{2}\right)  \right]  =t_{1}\wedge t_{2}-t_{1}t_{2}$.
Theorem 3.1 of \citet{lo2017functional} implies that%
\[
Cov\left(  H_{j,p}\left(  U\right)  ,H_{j,p^{\prime}}\left(  U\right)
\right)  =\int_{0}^{1}\int_{0}^{1}h_{j,p}\left(  t_{1}\right)  h_{j,p'}\left(
t_{2}\right)  \left(  t_{1}\wedge t_{2}-t_{1}t_{2}\right)  \mathrm{d}%
t_{1}\mathrm{d}t_{2}.
\]
Then it follows that
\begin{align*}
E\left[  \mathcal{V}_{j}\left(  p\right)  \mathcal{V}_{j}(  p^{\prime
})  \right]   &  =Cov\left(  \mathcal{V}_{j}\left(  p\right)
,\mathcal{V}_{j}(  p^{\prime})  \right)  =Cov\left(  H_{j,p}\left(
F_{j}\left(  X^{j}\right)  \right)  ,H_{j,p^{\prime}}\left(  F_{j}\left(
X^{j}\right)  \right)  \right)  \\
&  =Cov\left(  Q_{j}\left(  p\right)  \wedge X^{j},Q_{j}(  p^{\prime
})  \wedge X^{j}\right)  .
\end{align*}
\end{proof}

\begin{proof}[Proof of Lemma \ref{lemma.weak convergence F}]
The results follow from \eqref{eq.weak convergence phi^u}, \eqref{eq.weak convergence phi^d}, Assumption \ref{ass.Hadamard directional differentiability}, and Theorem 2.1 of \citet{fang2014inference}.
\end{proof}

\begin{proof}[Proof of Lemma \ref{lemma.bootstrap phi asymptotic distribution}]
By Lemma 5.2 of \citet{Beare2017improved}, Theorem 2.9 of \citet{kosorok2008introduction}, and Lemma 7.4 of \citet{K17} (which holds under Assumption \ref{ass.distribution}), we have that 
\[
\left(
\begin{array}
[c]{c}%
n_{1}^{1/2}(  \hat{Q}_{1}^{\ast}-\hat{Q}_{1})  \\
n_{2}^{1/2}(  \hat{Q}_{2}^{\ast}-\hat{Q}_{2})
\end{array}
\right)  \overset{\mathbb{P}}{\leadsto} \left(
\begin{array}
[c]{c}%
-Q_{1}^{\prime}\cdot\mathcal{B}_{1}\\
-Q_{2}^{\prime}\cdot\mathcal{B}_{2}%
\end{array}
\right)  \text{ in }L^{1}\left(  [0,1]\right)  \times L^{1}\left(  [0,1]\right)  .
\]
The result follows from Proposition 10.7 of \citet{kosorok2008introduction} (a conditional version of continuous mapping theorem).
\end{proof}

\begin{proof}[Proof of Lemma \ref{lemma.SI}]
As shown in the proof of Proposition 3.2 of \citet{Beare2017improved}, $\widehat{\mathcal{S}'_{\phi_m^w}}$ and $\widehat{\mathcal{I}'_{\phi_m^w}}$ are both Lipschitz continuous. 

For all sets $A,B\subset\left[  0,1\right]  $, define $\tilde{d}\left(
A,B\right)  =\sup_{a\in A}\inf_{b\in B}\left\vert a-b\right\vert $ and
\[
d\left(  A,B\right)  =\max\left\{  \tilde{d}\left(  A,B\right)  ,\tilde{d}\left(
B,A\right)  \right\}  .
\]
For every $\varepsilon>0$, we have that
\begin{align*}
\mathbb{P}\left(  \tilde{d}\left(  B_{0}\left(  \phi_{m}^{w}\right)
,\widehat{B_{0}\left(  \phi_{m}^{w}\right)  }\right)  >\varepsilon\right)   &
\leq\mathbb{P}\left(  B_{0}\left(  \phi_{m}^{w}\right)  \setminus
\widehat{B_{0}\left(  \phi_{m}^{w}\right)  }\neq\varnothing\right)  \\
&  \leq\mathbb{P}\left(  \sup_{p\in B_{0}\left(  \phi_{m}^{w}\right)
\setminus\widehat{B_{0}\left(  \phi_{m}^{w}\right)  }}\left\vert T_{n}%
^{1/2}\left\{  \hat{\phi}_{m}^{w}\left(  p\right)  -\phi_{m}^{w}\left(
p\right)  \right\}  \right\vert >\tau_{n}\xi^{1/2}
\right)  \\
&  \leq\mathbb{P}\left(  \sup_{p\in\left[  0,1\right]  }\left\vert T_{n}%
^{1/2}\left\{  \hat{\phi}_{m}^{w}\left(  p\right)  -\phi_{m}^{w}\left(
p\right)  \right\}  \right\vert >\tau_{n}\xi^{1/2}\right)  .
\end{align*}
Since $T_{n}^{1/2}( \hat{\phi}_{m}^{w}-\phi_{m}^{w})
\leadsto\mathbb{G}_{m}^{w}$ and $\tau_{n}\rightarrow\infty$, by continuous
mapping theorem and Slutsky's lemma,
\[
\tau_{n}^{-1}\sup_{p\in\left[  0,1\right]  }\left\vert T_{n}^{1/2}\left\{
\hat{\phi}_{m}^{w}\left(  p\right)  -\phi_{m}^{w}\left(  p\right)  \right\}
\right\vert \rightarrow_{p}0.
\]
This implies
\[
\mathbb{P}\left(  \tilde{d}\left(  B_{0}\left(  \phi_{m}^{w}\right)
,\widehat{B_{0}\left(  \phi_{m}^{w}\right)  }\right)  >\varepsilon\right)
\rightarrow0.
\]
For every $p\in\mathbb{R}  $ and $A\subset\mathbb{R}  $,
define $e\left(  p,A\right)  =\inf_{a\in A}\left\vert p-a\right\vert $. For
every $\varepsilon>0$, define
\[
D_{\varepsilon}=\left\{  p\in\left[  0,1\right]  :e\left(  p,B_{0}\left(
\phi_{m}^{w}\right)  \right)  \geq\varepsilon\right\}  .
\]
Suppose $D_{\varepsilon}\neq\varnothing$ and there is a sequence $p_{n}\in
D_{\varepsilon}$ such that $p_{n}\rightarrow p_{0}\in\mathbb{R}$. Then we have that
\begin{align*}
e\left(  p_{0},B_{0}\left(  \phi_{m}^{w}\right)  \right)   &  =\inf
_{p^{\prime}\in B_{0}\left(  \phi_{m}^{w}\right)  }\left\vert p_{0}-p^{\prime
}\right\vert =\inf_{p^{\prime}\in B_{0}\left(  \phi_{m}^{w}\right)
}\left\vert p_{0}-p_{n}+p_{n}-p^{\prime}\right\vert \\
&  \geq\inf_{p^{\prime}\in B_{0}\left(  \phi_{m}^{w}\right)  }\left\vert
p_{n}-p^{\prime}\right\vert -\left\vert p_{0}-p_{n}\right\vert \geq
\varepsilon-\left\vert p_{0}-p_{n}\right\vert \rightarrow\varepsilon,
\end{align*}
which implies $p_{0}\in D_{\varepsilon}$. Thus, $D_{\varepsilon}$ is closed and
compact. Since $\phi_{m}^{w}$ is continuous, then for every $\varepsilon>0$,
there is some $\delta_{\varepsilon}>0$ such that $\inf_{p\in D_{\varepsilon}%
}\left\vert \phi_{m}^{w}\left(  p\right)  \right\vert >\delta_{\varepsilon}$.
Now fix $\varepsilon>0$. We have that
\begin{align*}
&\mathbb{P}\left(  \tilde{d}\left(  \widehat{B_{0}\left(  \phi_{m}^{w}\right)
},B_{0}\left(  \phi_{m}^{w}\right)  \right)  >\varepsilon\right)   
=\mathbb{P}\left(  \sup_{p_{1}\in\widehat{B_{0}\left(  \phi_{m}^{w}\right)  }%
}\inf_{p_{2}\in B_{0}\left(  \phi_{m}^{w}\right)  }\left\vert p_{1}%
-p_{2}\right\vert >\varepsilon\right)  \\
&  \leq\mathbb{P}\left(  \sup_{p\in\widehat{B_{0}\left(  \phi_{m}^{w}\right)
}\setminus B_{0}\left(  \phi_{m}^{w}\right)  }\left\vert \phi_{m}^{w}\left(
p\right)  \right\vert >\delta_{\varepsilon},\sup_{p\in\widehat{B_{0}\left(
\phi_{m}^{w}\right)  }\setminus B_{0}\left(  \phi_{m}^{w}\right)  }\left\vert
T_{n}^{1/2}\hat{\phi}_{m}^{w}\left(  p\right) /\hat
{v}_{m}^{w}\left(  p\right) \right\vert \leq\tau_{n}  \right)  .
\end{align*}
By arguments similar to those in the proof of Lemma A.2 in \citet{Beare2017improved}, it can be
shown that there is some $M>0$ such that $\limsup_{n\to\infty}\sup_{p\in\left[  0,1\right]  }\hat{v}_{m}^{w}\left(  p\right)  <M$ almost surely, which implies 
$\mathbb{P}(\sup_{p\in\left[  0,1\right]  }\hat{v}_{m}^{w}\left(  p\right)  <M)\to 1$.
Let $A_n=\{\sup_{p\in\left[  0,1\right]  }\hat{v}%
_{m}^{w}\left(  p\right)    <M\}$ and 
\[
B_{n}=\left\{
\begin{array}
[c]{c}%
\sup_{p\in\widehat{B_{0}\left(  \phi_{m}^{w}\right)  }\setminus B_{0}\left(
\phi_{m}^{w}\right)  }\left\vert \phi_{m}^{w}\left(  p\right)  \right\vert
-\frac{\delta_{\varepsilon}}{2}\leq\sup_{p\in\widehat{B_{0}\left(  \phi
_{m}^{w}\right)  }\setminus B_{0}\left(  \phi_{m}^{w}\right)  }\left\vert
\hat{\phi}_{m}^{w}\left(  p\right)  \right\vert \\
\leq\sup_{p\in\widehat{B_{0}\left(  \phi_{m}^{w}\right)  }\setminus
B_{0}\left(  \phi_{m}^{w}\right)  }\left\vert \phi_{m}^{w}\left(  p\right)
\right\vert +\frac{\delta_{\varepsilon}}{2}
\end{array}
\right\} .
\]
By Lemma 3 of \citet{BDB14} and the strong law of large numbers, we have that
\[
\sup_{p\in\left[  0,1\right]  }\left\vert \hat{\Lambda}_{j}^{2}\left(
p\right)  -\Lambda_{j}^{2}\left(  p\right)  \right\vert \rightarrow0
\]
almost surely. This implies that $\sup_{p\in\left[  0,1\right]  }|\hat{\phi
}_{m}^{w}\left(  p\right)  -\phi_{m}^{w}\left(  p\right)  |\rightarrow0$
almost surely and $\mathbb{P}(A_{n}\cap B_n)\rightarrow1$. Let $C_n=A_n\cap B_n$. We then have
\begin{align*}
&  \mathbb{P}\left(  \sup_{p\in\widehat{B_{0}\left(  \phi_{m}^{w}\right)
}\setminus B_{0}\left(  \phi_{m}^{w}\right)  }\left\vert \phi_{m}^{w}\left(
p\right)  \right\vert >\delta_{\varepsilon},\sup_{p\in\widehat{B_{0}\left(
\phi_{m}^{w}\right)  }\setminus B_{0}\left(  \phi_{m}^{w}\right)  }\left\vert
T_{n}^{1/2}\hat{\phi}_{m}^{w}\left(  p\right)/\hat
{v}_{m}^{w}\left(  p\right)  \right\vert \leq\tau_{n}  \right)  \\
\leq&  \,\mathbb{P}\left(  \sup_{p\in\widehat{B_{0}\left(  \phi_{m}^{w}\right)
}\setminus B_{0}\left(  \phi_{m}^{w}\right)  }\left\vert \phi_{m}^{w}\left(
p\right)  \right\vert >\delta_{\varepsilon},\sup_{p\in\widehat{B_{0}\left(
\phi_{m}^{w}\right)  }\setminus B_{0}\left(  \phi_{m}^{w}\right)  }\left\vert
T_{n}^{1/2}\hat{\phi}_{m}^{w}\left(  p\right) /\hat
{v}_{m}^{w}\left(  p\right) \right\vert \leq\tau_{n}  ,C_{n}\right)  +\mathbb{P}\left(  C_{n}%
^{c}\right)  \\
  \leq&\,\mathbb{P}\left(  \sup_{p\in\widehat{B_{0}\left(  \phi_{m}^{w}\right)
}\setminus B_{0}\left(  \phi_{m}^{w}\right)  }\left\vert \phi_{m}^{w}\left(
p\right)  \right\vert >\delta_{\varepsilon},\sup_{p\in\widehat{B_{0}\left(
\phi_{m}^{w}\right)  }\setminus B_{0}\left(  \phi_{m}^{w}\right)  }\left\vert
\phi_{m}^{w}\left(  p\right)  \right\vert \leq\frac{\delta_{\varepsilon}}%
{2}+T_{n}^{-1/2}\tau_{n}M  ,C_{n}\right)
  \\
& +\mathbb{P}\left(  C_{n}^{c}\right) \rightarrow0.
\end{align*}
This implies
\[
\mathbb{P}\left( \tilde d\left(  \widehat{B_{0}\left(  \phi_{m}^{w}\right)  }%
,B_{0}\left(  \phi_{m}^{w}\right)  \right)  >\varepsilon\right)  \rightarrow0.
\]
Thus, we conclude that $d(B_{0}\left(  \phi_{m}^{w}\right)  ,\widehat
{B_{0}\left(  \phi_{m}^{w}\right)  })\rightarrow0$ in probability.
With the above results, we have that for every $h\in C\left(  \left[
0,1\right]  \right)  $,
\[
\left\vert \widehat{\mathcal{S}_{\phi_{m}^{w}}^{\prime}}\left(  h\right)
-\mathcal{S}_{\phi_{m}^{w}}^{\prime}\left(  h\right)  \right\vert \leq
\sup_{p,p^{\prime}\in\left[  0,1\right]  ,\left\vert p-p^{\prime}\right\vert
\leq d\left(  \widehat{B_{0}\left(  \phi_{m}^{w}\right)  },B_{0}\left(
\phi_{m}^{w}\right)  \right)  }\left\vert h\left(  p\right)  -h(p^{\prime
})\right\vert .
\]
Since $h$ is uniformly continuous on $\left[  0,1\right]  $, for every
$\varepsilon>0$, there is $\delta>0$ such that \\$\left\vert h\left(  p\right)
-h(p^{\prime})\right\vert \leq\varepsilon$ for all $p,p^{\prime}$ with
$|p-p^{\prime}|\leq\delta$. Thus,
\[
\mathbb{P}\left(  \left\vert \widehat{\mathcal{S}_{\phi_{m}^{w}}^{\prime}%
}\left(  h\right)  -\mathcal{S}_{\phi_{m}^{w}}^{\prime}\left(  h\right)
\right\vert >\varepsilon\right)  \leq\mathbb{P}\left(  d\left(  \widehat
{B_{0}\left(  \phi_{m}^{w}\right)  },B_{0}\left(  \phi_{m}^{w}\right)
\right)  >\delta\right)  \rightarrow0.
\]
By Lemma S.3.6 of \citet{fang2014inference}, $\widehat{\mathcal{S}_{\phi_{m}^{w}}^{\prime}}$
satisfies Assumption \ref{ass.FS}.

Let $\mu$ denote the Lebesgue measure. Then as shown above, for every $\varepsilon>0$,
\[
\mathbb{P}\left(  \mu\left(  B_{0}\left(  \phi_{m}^{w}\right)  \setminus
\widehat{B_{0}\left(  \phi_{m}^{w}\right)  }\right)  >\varepsilon\right)
\leq\mathbb{P}\left(  B_{0}\left(  \phi_{m}^{w}\right)  \setminus
\widehat{B_{0}\left(  \phi_{m}^{w}\right)  }\neq\varnothing\right)
\rightarrow0.
\]
Also, by definition,
\begin{align*}
\widehat{B_{0}\left(  \phi_{m}^{w}\right)  }\setminus B_{0}\left(  \phi
_{m}^{w}\right)   &  =\left\{  p\in\left[  0,1\right]  :\left\vert \hat{\phi
}_{m}^{w}(p)\right\vert \leq T_{n}^{-1/2}\tau_{n}\hat{v}_{m}^{w}\left(
p\right)  ,\left\vert \phi_{m}^{w}(p)\right\vert >0\right\}   .
\end{align*}
Let $A=\{\limsup_{n\to\infty}\{\sup_{p\in\left[  0,1\right]  }\hat{v}_{m}^{w}\left(  p\right)\}  <M,\sup_{p\in\left[  0,1\right]  }|\hat{\phi}%
_{m}^{w}(p)-\phi_{m}^{w}(p)|\rightarrow0\}$. Fix $\omega\in A$. 
We have that
\[
\sup_{p\in\left[  0,1\right]  }\left\vert \left\vert \phi_{m}^{w}\left(p\right)  \right\vert -\vert \hat{\phi}_{m}^{w}\left(  p\right)  \vert \right\vert \leq\sup_{p\in\left[0,1\right]  }|\hat{\phi}_{m}^{w}(p)-\phi_{m}^{w}(p)|\rightarrow0.
\]
Then we can find a sequence $\left\{  t_{n}\left(  \omega\right)  \right\}  $
such that $t_{n}\left(  \omega\right)  \downarrow0$ and for all $p\in[0,1]$,
\[
\left\vert \phi_{m}^{w}\left(  p\right)
\right\vert \leq\left\vert \hat{\phi}_{m}%
^{w}\left(  p\right)  \right\vert +t_{n}\left(  \omega\right)  .
\]
For such $\omega$, when $n$ is sufficiently large, it follows that for every $p\in\widehat{B_{0}\left(  \phi_{m}^{w}\right)  }\setminus B_{0}\left(  \phi
_{m}^{w}\right)$,
\[
\left\vert \phi_{m}^{w}\left(  p\right)
\right\vert \leq T_{n}^{-1/2}\tau_{n}\sup_{p\in\left[  0,1\right]  }\hat
{v}_{m}^{w}\left(  p\right)  +t_{n}\left(  \omega\right)  \leq T_{n}%
^{-1/2}\tau_{n}M+t_{n}\left(  \omega\right)  .
\]
Thus, we have that for such $\omega$, when $n$ is sufficiently large,
\[
\widehat{B_{0}\left(  \phi_{m}^{w}\right)  }\setminus B_{0}\left(  \phi
_{m}^{w}\right)  \subset\left\{  p\in\left[  0,1\right]  :0<\left\vert
\phi_{m}^{w}(p)\right\vert \leq t_{n}\left(  \omega\right)  +T_{n}^{-1/2}
\tau_{n}M\right\}  \equiv B_{n}\left(  \omega\right)  .
\]
Clearly, $B_{n}\left(  \omega\right)  \supset B_{n+1}\left(  \omega\right)
\supset\cdots$. It then follows that
\[
\lim_{n\rightarrow\infty}\mu\left(  B_{n}\left(  \omega\right)  \right)
=\mu\left(  \cap_{n=1}^{\infty}B_{n}\left(  \omega\right)  \right)  =0.
\]
This implies that for every $\omega\in A$,
\[
\mu\left(  \widehat{B_{0}\left(  \phi_{m}^{w}\right)  }\setminus B_{0}\left(
\phi_{m}^{w}\right)  \right)  \rightarrow0.
\]
As we have shown above, $\mathbb{P}\left(  A\right)  =1$. Thus, $\mu
(\widehat{B_{0}\left(  \phi_{m}^{w}\right)  }\setminus B_{0}\left(  \phi
_{m}^{w}\right)  )\rightarrow0$ almost surely. Now we have that under
$\mathrm{H}_{0}$ ($\tilde{\mathrm{H}}_{0}$), for every $h\in C\left(  \left[
0,1\right]  \right)  $,
\begin{align*}
\left\vert \widehat{\mathcal{I}_{\phi_{m}^{w}}^{\prime}}\left(  h\right)
-\mathcal{I}_{\phi_{m}^{w}}^{\prime}\left(  h\right)  \right\vert  & \leq
\max_{p\in\left[  0,1\right]  }\left\vert h\left(  p\right)  \right\vert
\cdot\left\{  \mu\left(  B_{0}\left(  \phi_{m}^{w}\right)  \setminus
\widehat{B_{0}\left(  \phi_{m}^{w}\right)  }\right)  +\mu\left(
\widehat{B_{0}\left(  \phi_{m}^{w}\right)  }\setminus B_{0}\left(  \phi
_{m}^{w}\right)  \right)  \right\}  \\
& \rightarrow_{p}0.
\end{align*}
By Lemma S.3.6 of \citet{fang2014inference}, $\widehat{\mathcal{I}_{\phi_{m}^{w}}^{\prime}}$
satisfies Assumption \ref{ass.FS}.
\end{proof}

\begin{proof}[Proof of Lemma \ref{lemma.bsconsistent}]
With Assumptions \ref{ass.Hadamard directional differentiability} and \ref{ass.FS}, \eqref{eq.weak convergence phi^u} and \eqref{eq.weak convergence phi^d}, and Lemma \ref{lemma.bootstrap phi asymptotic distribution}, the result follows from Theorem 3.2 of \citet{fang2014inference}. 
\end{proof}

\begin{proof}[Proof of Proposition \ref{prop.test}]
(i). Under the assumptions, the result follows from Theorem 3.3 of \citet{fang2014inference}. (ii). Under $\mathrm{H}_1$ ($\tilde{\mathrm{H}}_1$), by Theorem 1.3.6 of \citet{van1996weak}, $\mathcal{F}(\hat{\phi}_m^w)\to\mathcal{F}({\phi}_m^w)>0$ in probability. 
We define 
\begin{align}\label{eq.c tilde}
\tilde{c}_{m1-\alpha}^w=\inf\left\{c\in\mathbb R: \mathbb{P}\left(   {\mathcal F}(T_n^{1/2}(\hat\phi_m^{w\ast}-\hat\phi_m^w))\le c\,\middle|\,   \{  X_{i}^{1}\}  _{i=1}^{n_{1}}  ,\{  X_{i}^{2}\}  _{i=1}^{n_{2}}      \right) \ge 1-\alpha   \right\}.
\end{align} 
By similar arguments in the proof of Theorem S.1.1 of \citet{fang2014inference}, $\tilde{c}_{m1-\alpha}^w$ converges in probability to the $1-\alpha$ quantile of $\mathcal{F}(\mathbb{G}_m^w)$. Since by definition $\hat{c}_{m1-\alpha}^w\le\tilde{c}_{m1-\alpha}^w$, then $\mathcal{F}(\hat{\phi}_m^w)-T_n^{-1/2}\hat{c}_{m1-\alpha}^w$ converges to $\mathcal{F}({\phi}_m^w)$ in probability.
\end{proof}

\bibliographystyle{apalike}
\bibliography{reference1}
\end{document}